\documentclass{sig-alternate}

\usepackage{algorithmicx}
\usepackage[ruled]{algorithm}
\usepackage{algcompatible}
\usepackage[noend]{algpseudocode}

\usepackage{mathtools}
\DeclarePairedDelimiter{\ceil}{\lceil}{\rceil}

\usepackage{graphicx}
\usepackage{balance}  
\usepackage{url}
\usepackage{booktabs}
\usepackage{color}
\usepackage{xspace}
\usepackage{listings}
\usepackage{synttree}
\usepackage{enumitem}
\usepackage{subfigure}
\usepackage{comment}

\setlist[itemize]{noitemsep,nolistsep}
\setlist[enumerate]{noitemsep,nolistsep}
\usepackage[amsmath]{ntheorem}

\newtheorem{theorem}{Theorem}
\newtheorem{lemma}[theorem]{Lemma}
\newtheorem{proposition}[theorem]{Proposition}
\newtheorem{corollary}[theorem]{Corollary}

\theorembodyfont{\upshape}
\theoremstyle{definition}
\newtheorem{definition}[theorem]{Definition}
\newtheorem{example}[theorem]{Example}

\theoremstyle{remark}

\makeatletter
\newcommand{\shorteq}{%
  \settowidth{\@tempdima}{x}
  \resizebox{\@tempdima}{\height}{=}
}

\newcommand{\mypara}[1]{\vspace{0.5ex}\noindent\textbf{#1.}}
\newcommand{\myparaND}[1]{\vspace{0.5ex}\noindent\textbf{#1}}
\newcommand{\myparaNN}{\vspace{0.5ex}\noindent}
\newcommand{\TODO}[1]{\relax}
\newcommand{\set}[1]{\ensuremath{\{#1\}}}
\newcommand{\st}{\ensuremath{\,|\,}}
\newcommand{\directed}[1]{\ensuremath{{#1}^\star}}

\newcommand{\Figref}[1]{Fig.~\ref{#1}\xspace}
\newcommand{\Algref}[1]{Algorithm~\ref{#1}\xspace}
\newcommand{\appref}[1]{Apx.~\ref{#1}\xspace}
\newcommand{\Secref}[1]{Section~\ref{#1}\xspace}
\newcommand{\lemmaref}[1]{Lemma~\ref{#1}\xspace}
\newcommand{\propref}[1]{Prop.~\ref{#1}\xspace}
\newcommand{\theoremref}[1]{Theorem~\ref{#1}\xspace}
\newcommand{\defref}[1]{Def.~\ref{#1}\xspace}

\newcommand{\arity}{\ensuremath{a}}
\newcommand{\lftj}{\ensuremath{\textnormal{LFTJ}}\xspace}
\newcommand{\lftjtri}{\ensuremath{\textnormal{LFTJ-}\Delta}\xspace}
\newcommand{\lftjtribf}{\ensuremath{\textnormal{\textbf{LFTJ-}}\mathbf{\Delta}}\xspace}
\newcommand{\interval}[2]{\ensuremath{#1\cdot\!\cdot\!\cdot#2}}
\newcommand{\boxInterval}[6]{\ensuremath{[\interval{#1}{#2},\interval{#3}{#4},\interval{#5}{#6}]}}
\newcommand{\defEq}{\ensuremath{:\!\!\shorteq}\,}

\newcommand{\maxMem}{\ensuremath{\textnormal{mem}_{max}}}
\newcommand{\Tuple}{\ensuremath{n\mathsf{\textnormal{-}\!Tuple}}\xspace}
\newcommand{\TupleK}{\ensuremath{k\mathsf{\textnormal{-}\!Tuple}}\xspace}
\newcommand{\TupleM}{\ensuremath{m\mathsf{\textnormal{-}\!Tuple}}\xspace}

\newcommand{\ArrayOf}{\ensuremath{\mathsf{Array\;of}}\xspace}
\newcommand{\Int}{\ensuremath{\mathsf{int}}\xspace}
\newcommand{\AtomSet}{\ensuremath{\mathsf{AtomSet}}\xspace}
\newcommand{\SliceSet}{\ensuremath{\mathsf{SliceSet}}\xspace}
\newcommand{\LinIters}{\ensuremath{\mathsf{LinearIterators}}\xspace}
\newcommand{\LeapFrogJoinSet}{\ensuremath{\mathsf{Set\,of\,LeapfrogJoins}}\xspace}
\newcommand{\Bool}{\ensuremath{\mathsf{bool}}\xspace}

\newcommand{\gLow}{\ensuremath{\textnormal{low}}\xspace}
\newcommand{\gHigh}{\ensuremath{\textnormal{high}}\xspace}

\newcommand{\gLowAt}[1]{\ensuremath{\textnormal{low}[#1]}\xspace}
\newcommand{\gHighAt}[1]{\ensuremath{\textnormal{high}[#1]}\xspace}

\newcommand{\Iters}[1]{\ensuremath{\textnormal{Iters[#1]}}\xspace}
\newcommand{\Lfs}[1]{\ensuremath{\textnormal{Lfs[#1]}}\xspace}
\newcommand{\AtEnd}{\ensuremath{\textnormal{atEnd}}\xspace}
\newcommand{\MaxVal}{\ensuremath{\textnormal{max\_value}}\xspace}
\newcommand{\MinVal}{\ensuremath{\textnormal{min\_value}}\xspace}

\newcommand{\True}{\ensuremath{\textnormal{\texttt{true}}}\xspace}
\newcommand{\False}{\ensuremath{\textnormal{\texttt{false}}}\xspace}

\newcommand{\PVar}{\ensuremath{\textnormal{$i$}}}
\newcommand{\DVar}{\ensuremath{\textnormal{$d$}}}
\newcommand{\IterVar}{\ensuremath{\textnormal{$iter$}}}

\newcommand{\NegInfTuple}{\ensuremath{[-\infty,...,-\infty]}\xspace}
\newcommand{\InfTuple}{\ensuremath{[\infty,...,\infty]}\xspace}
\newcommand{\Gets}{\ensuremath{\gets}\xspace}
\newcommand{\Inf}{\ensuremath{\infty}\xspace}
\newcommand{\NegInf}{\ensuremath{-\infty}\xspace}

\newcommand{\Successor}[1]{\ensuremath{\mathsf{succ}(#1)}}

\newcommand{\BudgetVar}[1]{\ensuremath{\textnormal{budget}[#1]}\xspace}
\newcommand{\Preds}[1]{\ensuremath{\textnormal{atoms}[#1]}\xspace}
\newcommand{\PredsT}[1]{\ensuremath{\textnormal{\texttt{atoms}}[#1]}\xspace}
\newcommand{\SpillsT}[1]{\ensuremath{\textnormal{\texttt{spill}}[#1]}\xspace}
\newcommand{\Spills}[1]{\ensuremath{\textnormal{spill}[#1]}\xspace}
\newcommand{\Slices}[1]{\ensuremath{\textnormal{S}[#1]}\xspace}

\newcommand{\Budget}{\ensuremath{\textnormal{mem}}\xspace}
\newcommand{\UsedMemAccu}{\ensuremath{\textnormal{usedMem}}\xspace}
\newcommand{\LocalSlice}{\ensuremath{\textnormal{slice}}\xspace}
\newcommand{\LocalMem}{\ensuremath{\textnormal{m}}\xspace}
\newcommand{\LocalAtoms}{\ensuremath{\textnormal{atms}}\xspace}

\newcommand{\Var}{\ensuremath{\textnormal{$i$}}\xspace}
\newcommand{\LeftOverMem}{\ensuremath{\textnormal{leftoverMem}}\xspace}
\newcommand{\LeftOverMemT}{\ensuremath{\textnormal{\texttt{leftoverMem}}}\xspace}

\newcommand{\LoopNums}[1]{\ensuremath{L_{#1}}\xspace}
\newcommand{\BoxNums}[1]{\ensuremath{B_{#1}}\xspace}

\newcommand{\Method}[1]{\Call{#1}{}\xspace}
\newcommand{\Nothing}{\,$\!$}
\newcommand{\CallParens}[1]{\Call{#1}{\Nothing}\xspace}

\newcommand{\BoxLevel}{BoxUp\xspace}
\newcommand{\TrieSlice}[4]{\ensuremath{{#1}^{#2}_{#3\rightarrow #4}}}

\relax
\boilerplate={\raisebox{0.7em}{$\star$} 
The following is ``work-in-progress'', yet has some promising results I am happy to share.
The results I am most proud of are: (1) Achieved performance in absolute numbers on large datasets. It is very exciting that these are achieved \emph{not} by a specialized triangle-listing algorithm but by a \emph{general-purpose join} algorithm. (2) The theoretical results about the in-memory LFTJ and its optimiality for input instance classes characterized by their arboricity (Thms.\,17 \& 19). (3) The simplicity of the proposed out-of-core technique and the elegance for subsetting TrieArrays. While the achieved I/O complexity for the out-of-core technique can likely be improved, the experiments show that this is not the bottleneck.}
\copyrightetc{}

\begin{document}
\conferenceinfo{}{}
%
%

%
\title{General-Purpose Join Algorithms\\for Listing Triangles in Large Graphs}
\subtitle{Early Research Report\raisebox{0.7em}{${\star}$}}

\numberofauthors{1} 
\author{
\alignauthor
Daniel Zinn\\LogicBlox, Inc.
}
\date{30 July 1999}

\maketitle
\begin{abstract} 
We investigate applying general-purpose join algorithms to the triangle listing problem in an out-of-core context.
In particular, we focus on Leapfrog Triejoin (LFTJ) by Veldhuizen\cite{Veldhuizen14}, a recently proposed, worst-case optimal algorithm. 
We present ``boxing'': a novel, yet conceptually simple, approach for feeding input data to LFTJ. Our extensive analysis shows that this
approach is I/O efficient, being worst-case optimal (in a certain sense). Furthermore, if input data is only a constant factor larger than the available memory, then a boxed LFTJ essentially maintains the CPU data-complexity of the vanilla LFTJ. Next, focusing on LFTJ applied to the 
triangle query, we show that for many graphs boxed LFTJ matches the I/O complexity of the recently by Hu, Tao and Yufei proposed specialized 
algorithm MGT \cite{hu2013massive} for listing tiangles in an out-of-core setting. We also strengthen the analysis of LFTJ's computational complexity for the triangle query by considering families of input graphs that are characterized not only by the number of edges but also by a measure of their density. 
E.g., we show that LFTJ achieves a CPU complexity of $O(|E|\log|E|)$ for planar graphs, while on general graphs, no 
algorithm can be faster than $O(|E|^{1.5})$. Finally, we perform an experimental evaluation for the 
triangle listing problem confirming our theoretical results and showing the overall effectiveness of our approach.
On all our real-world and synthetic data sets (some of which containing more than 1.2 billion edges) 
LFTJ in single-threaded mode is within a factor of $3$ of the specialized MGT; a penalty that---as 
we demonstrate---can be alleviated by parallelization.
\end{abstract}

\newpage
\section{Introduction}

Hu, Tao, and Yufei \cite{hu2013massive} recently proposed a novel algorithm (MGT) for listing triangles in large graphs that is 
both I/O and CPU efficient; and also outperforms existing competitors by an order of magnitude.
At the same time, there has been exciting theoretical research that
shows it is possible to design so-called worst-case optimal join algorithms \cite{atserias2008size,Veldhuizen14,ngo2014beyond,ngo2012worst}. 
This begs the question: How would general-purpose join algorithms compare
to the best specialized triangle-listing algorithms in a setting where not all data 
fits into main memory? 

This question is motivated by the desire of building \emph{general-purpose} systems that can empower
their (domain) users to pose and run queries in a \emph{declarative} and general language, such as SQL
or Datalog---a goal that likely is little controversial. We focus on the out-of-core setting not only  
because of the obvious reasons of input or intermediary data not fitting in main memory, but also 
because we like to utilize graphics processing cards (GPUs) as high-throughput co-processors 
during query evaluation. GPU memory
is currently limited to up to around 12GB \cite{NVIDIADevZone}, highlighting the urgency for robust out-of-core techniques.

The triangle listing is the basic building block for many other graph algorithms and key ingredient for graph metrics such as triangular clustering, finding cohesive subgraphs etc \cite{hu2013massive, raymond2002rascal, rhodes2003clip, palla2005uncovering}.
In addition, it has gotten extensive attention in the research literature among several fields: graph theory, databases and network analysis to name a few. Here, both in-memory as well as in an out-of-core algorithms have been studied.
Having a \emph{general-purpose} technique being able to compete with the best-in-class \emph{hand-crafted}
algorithms that are \emph{specific} for triangle listing, would indeed, be very good news for the database 
community advocating high-level, declarative query languages.

We selected Leapfrog-Triejoin (LFTJ) by Veldhuizen as the general-purpose join algorithm for our study. This is for various reasons: 
(1) its elegance allows efficient implementations with various optimizations, (2) by nature, LFTJ only uses
$O(1)$ intermediary data, making it a very good candidate in the out-of-core context, and (3) because of
its strong theoretical worst-case guarantees \cite{Veldhuizen14}. LFTJ's worst-case guarantee in its generality is technical \cite{Veldhuizen14}. Roughly, it guarantees that for a given query and input $I$, LFTJ will never perform asymptotically more steps (up to a log-factor) than what are strictly necessary for any correct algorithm on inputs $I'$ that are \emph{similar} to $I$. Here, similar means that, eg., the sizes of the input relations cannot change nor can certain other statistics of the data. 

\mypara{Model \& Assumptions} 
We restrict 
our attention to full-conjunctive queries, and use a Datalog syntax and terminology to describe queries (or joins).
Our formal setting is the standard one for considering I/O efficient algorithms: Input, intermediary and output
data can exceed the amount of available main memory $M$ (measured in words to store one atomic value), 
in which case it can be read (written) from (to) secondary
storage with the granularity of a block that has size $B$. Reading or writing a block incurs 1 unit of I/O cost.
For I/O and CPU cost, we consider data complexities, that is we assume the query to be fixed and small. In particular,
we like $M/B$ to be larger than, say 10 times, the number of atoms multiplied by their maximum arity. Furthermore,
to simplify complexity results, we assume that $|I|/B$ is larger than $\log|I|$. This restriction is mostly theoretic: Using a block size of 
64KiB with a 64-bit word-width, inputs only need to be larger than 15MiB to satisfy the requirement. 
With these assumptions in mind, we make the following contributions:
\subsubsection*{Contributions} 
\mypara{Boxing LFTJ}
        We present and analyze a novel strategy we call \emph{boxing} for out-of-core execution of a multi-predicate, 
        worst-case optimal join algorithm (Leapfrog-Triejoin). This method exhibits the following properties:
        
        \textbf{(1)} 
        For queries with $n$ variables, executing on input data $I$ and producing 
        output data of size $K$, boxed LFTJ requires at most $O(|I|^n/(M^{n-1} B) + K/B)$ block I/Os. We show that this bound is worst-case 
        optimal, in the sense that for any $n$, we can construct a query such that no algorithm can have an asymptotically better 
        bound with respect to $I$ and $K$. 

        \textbf{(2)} 
        We further show that if the input data exhibits limited skew (in the sense we will make precise) then boxed LFTJ requires 
        only $O(|I|^r/(M^{r-1} B) + K/B)$ I/Os. Here, $r$ denotes the \emph{rank} of the query---a property we will 
        define. The rank of a query never exceeds the number of variables used in the Datalog body, and is often lower.

        \textbf{(3)} 
        We also analyze the computational complexity of boxed \lftj. Here,
        we show that if the input size $|I|$ is only a constant factor larger than the
        available memory $M$, then the asymptotic CPU work performed by the boxed \lftj 
        (essentially)\footnote{Except when the in-memory \lftj's complexity is in $o(|I|)$, in which case the boxed
        version's complexity is $O(|I|)$.} matches the asymptotic complexity of the in-memory \lftj maintaining its
        theoretical guarantees.

\mypara{Boxed \lftjtribf}
        We apply boxing to the triangle-listing problem. 
        Here, the input graph exhibits limited-skew if the degree of its nodes 
        is limited by $M/9$. With 100GiB of main memory this allows graphs containing 
        nodes with up to 1.3 million neighbors. 

        On such graphs, our approach requires $O(|I|^2/(MB) + K/B)$ block I/Os, 
        matching the asymptotic I/O bound of the recently presented 
        specialized algorithm MGT \cite{hu2013massive} for triangle listing.

\mypara{In-memory \lftjtribf}
        We also tighten the analysis for the CPU complexity of the conventional in-memory LFTJ applied 
        to the triangle listing query with non-trivial arguments.
        It is easy to see that \lftjtri's achieved asymptotic complexity of $O(|E|^{1.5}\log|E|)$ is 
        worst-case optimal modulo the log-factor. We improve on this result in two ways:

        \textbf{(1)}
        We show that 
        for graphs $G=(V,E)$ with an arboricity $\alpha(|E|)$, \lftjtri requires
        $O(|E|\alpha(|E|)\log|E|)$ work. A graph's arboricity is a measure of its density (as we will explain later)
        which never exceeds $O(\sqrt{|E|})$. 
        Moreover, $\alpha$ is substantially smaller for many graphs
        \cite{chiba1985arboricity,lin2012arboricity};
        for example, for both planar graphs and graphs with fixed maximum degree, $\alpha \in O(1)$.
        As a corollary, we thus obtain that \lftjtri runs in $O(|E|\log|E|)$ on planar graphs.

        \textbf{(2)}
        We further improve on the worst-case optimality analysis: We show that even if we are only interested in
        \emph{families of graphs} for which their arboricity is limited by any function $\hat\alpha \in o(\sqrt{|E|})$, e.g.\,by $30\log|E|$,
        and we would like to design a specialized algorithm that (only) works (well) on these graphs, then this
        algorithm cannot have an asymptotic complexity that is in $o(\hat\alpha(|E|)|E|)$. This result shows
        that \lftjtri is worst-case optimal for any of these families (modulo the log-factor).

\mypara{Evalation} We further present an experimental evaluation, where we focus on the triangle query.
We confirm that the boxing technique works well, especially when the input data is only a constant factor larger than the 
available memory: on real-world and synthetic graphs with each more than 1.2 billion nodes, boxing only introduces little CPU 
overhead; and has good performance even when only limited main memory is available. We also compare the raw performance against
two competitors: a specialized \cite{schank2007algorithmic} C++-based implementation in the graph-processing system Graphlab 
\cite{Low2012GraphLab} and the specialized triangle listing algorithm MGT \cite{hu2013massive}. LFTJ is about 65\%
the Graphlab implementation, yet scales to larger data sizes. When running single-threaded, LFTJ is on average 3x slower than 
MGT. Our parallelized version of LFTJ, however, is slightly faster than the single-threaded MGT (about 30\%
main memory is restricted to as much as 10\%

The rest of the paper is structured as follows: \Secref{SecAssumptions} reviews the relevant background information.
We present and analyze the boxing strategy for LFTJ in \Secref{SecAutoBoxing}. \Secref{SecLFTJTriangle} analyzes the in-memory and the boxed variant of LFTJ applied to the triangle query. \Secref{SecImpl} highlights some important aspects of our implementation, before we experimentally evaluate our approaches in \Secref{SecExperiments}. We review related work in \Secref{SecRelatedWork} and conclude in \Secref{SecConclusion}. 

\section{Background} \label{SecAssumptions}
\subsection{Review: Leapfrog-Triejoin (LFTJ) \cite{Veldhuizen14}}

LFTJ \cite{Veldhuizen14} is a \emph{multi-predicate} join algorithm. Unlike traditional \emph{binary} join algorithms such as Hash-Join or Sort-Merge-Join which take two relations as input, LFTJ takes as input $n$ relations together with the join conditions. 

\begin{example}[LFTJ] 
Consider the query:
\[ 
   Q(x,y,z) \leftarrow R(x,y), S(x,z), T(y,z).
\]
With binary joins, we first join, e.g., $R(x,y)$ with $S(x,z)$ to obtain $RS(x,y,z)$ and then join $RS(x,y,z)$ with $T(y,z)$ to obtain $Q(x,y,z)$. LFTJ, on the other hand directly computes $Q(x,y,z)$ given \emph{all} predicates in the body of the rule as an input without storing \emph{any sizeable} intermediary results. 

Some notation is necessary: for a binary relation $R(x,y)$, let $R(x,\_)$ denote the set of values in the first column, i.e., $R$ projected to its first attribute; further let $R_a(y)$ denote the projection to the second attribute \emph{after} only selecting tuples that have the constant $a$ as the first attribute. 

LFTJ operates by first fixing an order of the variables occuring in the rule body. In our example, we might pick $x,y,z$ as the order. Then, LFTJ finds all possible values $a$ for the variable $x$. This is done by performing an intersection of $R(x,\_)$ and $S(x,\_)$, i.e., the first column of $R$ with the first column of $S$ because the variable $x$ occurs in these atoms. Now, as soon as the first of such $a$ is found, LFTJ is looking for values $b$ of $y$, the next variable in the variable-order. Here, LFTJ computes the intersection of $R_a(y)$ and $T(y,\_)$. 
Again, as soon as the first of such $b$ is found, LFTJ is looking for values $c$ for $z$ by computing the intersection of $S_a(z)$ and $T_b(z)$. If any of these $c$ is found LFTJ reports the tuple $(a,b,c)$ in the output. Once the intersection of $S_a(z)$ and $T_b(z)$ has been computed, LFTJ back-tracks its search to the variable $y$ and looks for the next $b'$. Back-tracking continues up to the first variable and LFTJ finishes when no new $a'$ can be found anymore. %
A key to LFTJ's performance is to efficiently compute the various intersections. This is achieved via the method of a Leapfrog join (LFJ) which, as we detail below, leverages that relations are pre-sorted.

\end{example}

\mypara{Trie representation for relations}
\begin{figure}
\centering
\mbox{%
          \subfigure[\label{FigTrieTrieArrayA}$R$]{ \includegraphics[scale=0.78]{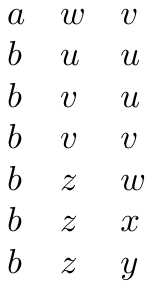} }
          \subfigure[\label{FigTrieTrieArrayB}Trie of $R$]{     \includegraphics[scale=0.78]{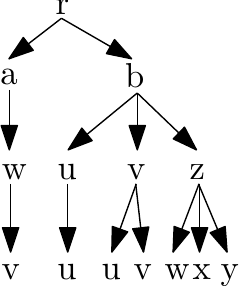} }
          \subfigure[\label{FigTrieTrieArrayC}TrieArray of $R$]{\includegraphics[scale=0.78]{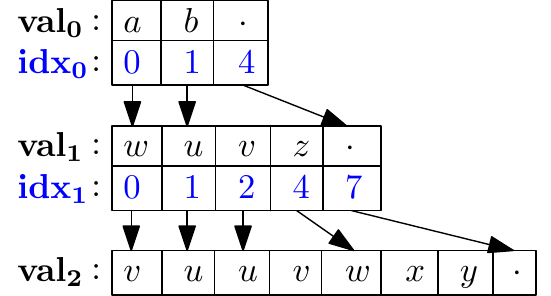} }
}
\vspace*{-4mm}
\caption{\label{FigTrieTrieArray} Trie and TrieArray of a ternary relation $R$}
\end{figure}
It is convenient to think of relations to be represented as a Trie\footnote{also called prefix tree, radix tree or digital tree}. A Trie is a tree that stores each tuple of a relation as a path from the root node to a child node. See \Figref{FigTrieTrieArrayA} for an example of a ternary relation with its trie in \Figref{FigTrieTrieArrayB}. In general, a Trie for a relation with arity $\arity$ has a height of $\arity$. For a relation $R(x_1,\dots,x_\arity)$, the nodes at height $i$ store values from the $i$th column of $R$. We require that children of the \emph{same} node $n$ are unique and ordered increasingly. For example in \Figref{FigTrieTrieArrayB} at level 2, the children of b are the values u, v, and w, which are in increasing order.

\mypara{TrieIterators} LFTJ accesses relational data not directly but via a \emph{TrieIterator} interface. This not only allows various storage schemes\footnote{e.g., regular B+-Trees, sorted list of tuples, or the TrieArrays we describe later} but also facilitates uniform handling of ``infinite'' predicates such as \texttt{Equal}, \texttt{SmallerThan} or \texttt{Plus}. The TrieIterator interface provides methods to navigate the Trie of a relation. It can be thought of as a pointer to a node in the Trie. The detailed methods for Trie navigation are given in \appref{AppTrieNavigation}. The methods are \Method{value()} to access a data value; \Method{open()} and \Method{close()} to move up and down in the trie. The \emph{linear iterator} methods \Method{next}, \Method{seek}, and \Method{atEnd} are used to move within \emph{unary} ``sub-relations'' $A$ such as $R(x,\_)$ or $R_a(y)$. Here, \Method{next} moves one step right and
\Method{seek()} is used to
forward-position the iterator to the element with value $v$; if $v$ is not in $A$, then the iterator is placed at the element with the smallest value $w > v$. In general, if the iterator passes the end of the represented relation such as $R_a(y)$, the \Method{atEnd} will return \texttt{true}. A key to good LFTJ performance is that back-end data structures efficiently support these TrieIterator operations. In fact, the theoretical guarantees given by LFTJ require that \Method{value}, \Method{key}, \Method{atEnd} have complexity $O(1)$. Furthermore, \Method{seek} and \Method{next} must not take longer than $O(\log N)$ individually and must have an armortized cost of at most $O(1+\log(N/m))$ if $m$ keys are visited. Here, $N$ stands for the size of the \emph{unary} relation the iterator is for; eg, $R_a(y)$.

\label{SecLFTJcomplexityrequirement}

\mypara{Leapfrog Join} A basic building block of Leapfrog Triejoin (LFTJ) is Leapfrog join (LFJ). It computes the intersection of multiple unary relations. For this, LFJ has a linear iterator for each of its input relations. An execution of LFJ is reminiscent of the merge-phase of a merge-sort; however instead of returning values that are in \emph{any} of the inputs, we search and return values that are in \emph{all} input relations. 
To do so efficiently, we use \Method{seek} to iteratively advance the iterator positioned at the relation with the smallest value to the largest value amongst the iterators. If all iterators are placed on the same value, we have found a value of the intersection. %

Using LFJ to join $n$ relations with $N_{\min}$ and $N_{\max}$ denoting the cardinalities of the smallest and largest relation, respectively, has the following complexity:
\begin{proposition}[3.1 in \cite{Veldhuizen14}]
The running time of Leapfrog join is $O(N_{\min}\log(N_{\max}/N_{\min}))$. 
\end{proposition}

The detailed algorithms for the Leapfrog join as well as LFTJ are given in \appref{AppLFTJAlgos} as reference; for an even more detailed introduction and reference see \cite{Veldhuizen14}.

\mypara{Leapfrog TrieJoin Restrictions} 
LFTJ requires that no variable occurs more than once in a single body atom. This can be achieved via simple rewrites: Given a join with, e.g., the atom $A = R(x,y,x)$ in the body, we introduce a new variable $x'$ and replace $A$ by $R(x,y,x'),\mathtt{Eq}(x,x')$ where $\mathtt{Eq}$ is the infinite equal-relation which itself is represented by a specialized TrieIterator.

As mentioned above, LFTJ is parameterized by an order on the variables of the join. This order is usually chosen by an optimizer as the exact order might influence runtime characteristics and can have an effect on the theoretical bounds for the I/O complexity as we will detail below. Furthermore, the chosen order determines the sort-order of the input relations: In particular, arguments in atoms of the join body must form a subsequence of the chosen order. E.g., consider the order $x,y,z$: body atoms $R(x,z)$ or $S(y)$  are allowed while the atom $T(y,x)$ needs to be replaced by an \emph{alternative index} $T_{2,1}(x,y)$ which is created as $T_{2,1}(x,y) \Gets T(y,x)$. These indexes are created in a pre-processing step.

\subsection{TrieArrays} 
We use a simple array-encoding for Tries, which is inspired by the Compressed-Sparse-Row (CSR) format---a commonly used format to store graphs. As an example see \Figref{FigTrieTrieArrayC} for the representation of the trie given in \Figref{FigTrieTrieArrayB}. The data values are stored in flat arrays called \emph{value}-arrays.
 Index arrays are used to separate children at the same tree level but from different parent nodes. An $n$-ary relation has $n$ value arrays and $n-1$ index arrays. In particular, the children of a node $n$ stored in the value array $\mathbf{val_i}$ at position $j$ are stored in the array $\mathbf{val_{i+1}}$ starting at the index from $\mathbf{idx_i}[j]$ until the index $\mathbf{idx_i}[j+1]$ inclusively. E.g. in \Figref{FigTrieTrieArrayC}, the children $w,x,y$ of $w$ from $\mathbf{val_1}[3]$ are stored in $\mathbf{val_2}$ from $\mathbf{idx_1}[3] = 4$ to $\mathbf{idx_1}[4] = 7$.

To reduce notation, we will often simply identify a relation $R$ with its TrieArray representation and vice versa in the rest of the paper. For example, when we write a $n$-ary TrieArray we mean a TrieArray for an $n$-ary relation $R$.

All TrieIterator operations are trivial to implement for TrieArrays; except possibly seek, where some attention needs to be given to achieve the
required armortized complexity. Here, instead of starting the binary in the middle of the remaining sub-array, we probe with an exponentially increasing lookup sequence of eg., $1,4,16,64,...$ to narrow lower and upper bounds for the binary search. 

While the TrieArray representation is beneficial for execution, it is also fairly cheap to create: 
\begin{proposition} The TrieArray representation of a relation $R$ requires no more than $O(|R|)$ space and can be built in $O(SORT(R))$ time and I/O complexity.
\end{proposition}
\textsc{Sketch.} The space requirement is obvious; furthermore the data structure can be built from a lexicographically sorted $R$ in two passes: pass 1 determines the sizes of the value and index arrays, pass 2 fills in data.
$\qed$

\subsection{LFTJ for Computing Triangles}

Given a simple, undirected graph $G$ and let $\directed{G} = (V,E)$ be its directed version, that is for each edge $\set{a,b} \in E(G)$, $E$ contains the pair  $(\min\set{a,b}, \max\set{a,b})$. The query 
\[
  T(x,y,z) \leftarrow {E}(x,y), {E}(x,z), {E}(y,z), x < y < z.
\]
computes all triangles in $\directed{G}$ of the form: 
          \includegraphics[scale=1]{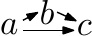}.
The output $T$ coincides with the triangles in $G$. %
Since $x<y<z$ is already implied by the atoms containing the edge relation, we can omit the
inequality from the query obtaining:
\begin{equation}
  T(x,y,z) \leftarrow {E}(x,y), {E}(x,z), {E}(y,z). \tag{\ensuremath{\Delta}}
  \label{EQtriangleQuery}
\end{equation}

\section{Boxing LFTJ} \label{SecOOC} \label{SecAutoBoxing}

We first motivate our strategy by showing that LFTJ can suffer from excessive I/O operations in an external-memory setting with a block-based least-recently-used memory replacement strategy. As example, we use the triangle query with specifically crafted input graphs.

\mypara{LFTJ on the triangle query}
It is useful to highlight the steps that LFTJ performs for the triangle query \eqref{EQtriangleQuery}. 
These are summarized in \Algref{AlgFLTJTriangle}. Note that \Algref{AlgFLTJTriangle} is \emph{not}
the pseudo-code of the program we use to list triangles; it only summarizes the steps LFTJ 
performs when run on the triangle query. First, the leapfrog join at level $x$ 
for the atoms $E(x,y)$ and $E(x,z)$ computes the intersection between $E(x,\_) = V_1$ and $E(x,\_) = V_1$. 
Then, for each found value $a$ for $x$, we perform a leapfrog-join at level $y$ computing the intersection of $E_a(y) = D(a)$ with $V_1 = E(y,\_)$, because the variable $y$ occurs in the atoms $E(x,y)$ and $E(y,z)$. In the last step, we find bindings for $z$
by intersecting $D(a) = E_a(z)$ with $D(b) = E_b(z)$ because $z$ occurs in the atoms $E(x,z)$ and $E(y,z)$.

\begin{algorithm}[t]
\caption{\label{AlgFLTJTriangle} Steps Performed by Leapfrog-Triejoin on the Triangles Query $T(x,y,z) \leftarrow E(x,y), E(x,z), E(y,z)$.}
\begin{algorithmic}[1]
   \For{$a \in V_1 \cap V_1$}               \Comment{$V_1 \defEq \set{x \st (x,v) \in E}$}
      \For{$b \in V_1 \cap D(a)$}  \Comment{$D(v) \defEq \set{x \st (v,x) \in E}$}
         \For{$c \in D(a) \cap D(b)$}
           \State \textbf{yield} $(a,b,c)$  \Comment{triangle found}
         \EndFor
      \EndFor
   \EndFor
\end{algorithmic}
\end{algorithm}

\mypara{Example inputs that causing excessive I/O} 
For $N \ge M+B$, consider the graph $G_N=(V,E)$ with edges $E$ as:
\[
   E = \set{ (x,y) \st x = 0,\dots,N \textnormal{ and } y = N - B(x \textnormal{ mod } T) }
\]
where $T = M/B + 1$ being slightly larger than the number blocks fitting into main memory at once.
See \Figref{FigThreshingRel} in the appendix for an example with $N=24$, $M=20$, $B=4$, and $T=6$. The key idea is that we place values in the second column of $E$ by $B$ apart which will cause LFTJ to perform an I/O for every tuple in $E$ for step 3 in \Algref{AlgFLTJTriangle}; furthermore, we make sure values in the second column repeat in groups large enough that loading all blocks in a group will preempt the first block from memory effectively prohibiting the algorithm to reuse the earlier loaded blocks. 

\begin{proposition} \label{PropIOs}
\lftjtri incurs at least $2|E(G_N)|$ I/Os for the above defined graph $G_N$ with a TrieArray data representation and a LRU memory replacement strategy. 
\end{proposition}
\begin{proof} See \appref{AppProofIOExample}.
\end{proof}

\subsection{High-Level Idea}

\begin{figure}
\centering
          \subfigure[\label{FigExampleGraphA}Graph $G$]{\includegraphics[scale=0.9]{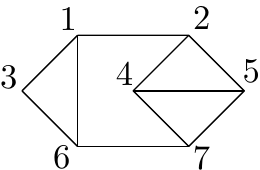} }
          \hfill
          \subfigure[\label{FigExampleGraphEdges}$E = E(\directed{G})$]{
\begin{minipage}[b][3em][l]{0.6\columnwidth}  
\resizebox*{0.9\columnwidth}{!}{$
\begin{array}{|r||c|c|c|c|c|c|c|c|c|c|c}
\hline
\mathbf{x} & 1 & 1 & 1 & 2 & 2 & 3 & 4 & 4 & 5 & 6 \\ \hline
\mathbf{y} & 2 & 3 & 6 & 4 & 5 & 6 & 5 & 7 & 7 & 7 \\
\hline
\end{array}
$}
\end{minipage}
          }
          \subfigure[\label{FigExampleGraphADirected}Dir.\,graph $\directed{G}$]{\includegraphics[scale=0.9]{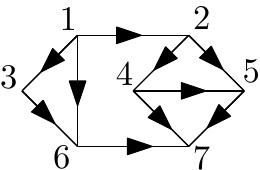} }
          \subfigure[\label{FigExampleGraphTrieArray}TrieArray $T$ for $E = E(\directed{G})$]{\includegraphics[scale=0.8]{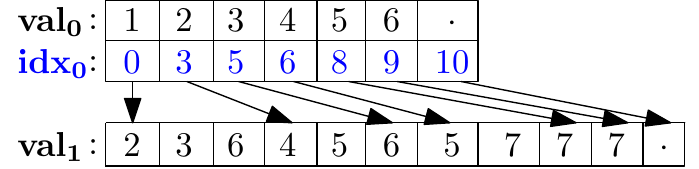} }
          \subfigure[\label{FigExampleGraphTiling}Boxed Search Space]{\includegraphics[width=0.525\columnwidth]{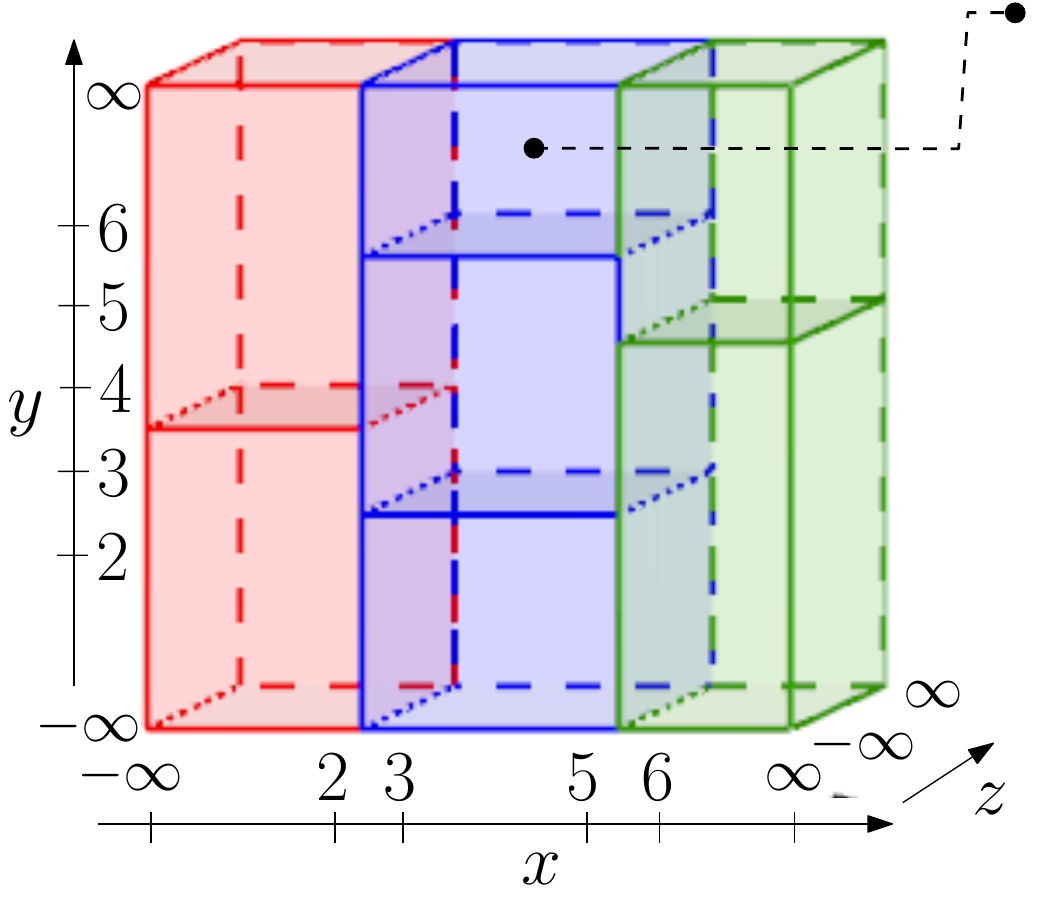}}
          \hspace{-3ex}
          \vspace{2ex}
          \subfigure[\label{FigExampleGraphTrieSlice}Example Box \& TrieSlices]{
\begin{minipage}[b][3.9cm][t]{0.47\columnwidth}  
{
{\textbf{Box} \boxInterval{3}{5}{6}{\infty}{-\infty}{\infty}}
\vfill
  \hspace{2ex}\includegraphics[scale=0.72]{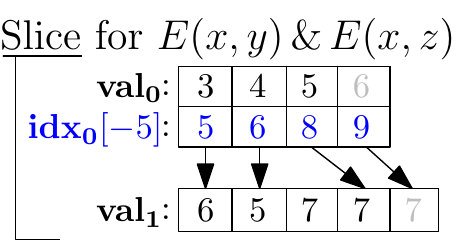}
\vspace{1ex}
\vfill
   \hspace{2ex}\includegraphics[scale=0.72]{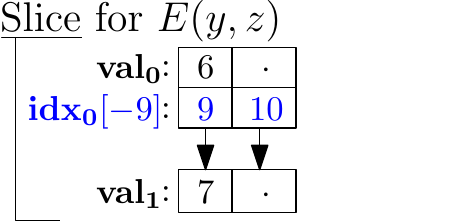}
\vspace{2ex}
}
\end{minipage}                  
}
\vspace*{-5mm}
\caption{\label{FigKeyFigure} Example for out-of-core technique for \lftjtri, i.e. $T(x,y,z)\!\!\leftarrow\!\!E(x,y),E(x,z),E(y,z)$ on $E(\directed{G})$}
\end{figure}

We now describe our out-of-core adaptation for LFTJ. 
LFTJ with a variable order $x_1,\dots,x_n$ computes the join by essentially searching over an $n$ dimensional space in which each dimension $i$ spans over the domain of the variable $x_i$. Loosely speaking, the space is searched in lexicographical order. As the example above demonstrates, this can lead to excessive I/O costs. Further I/O accesses are caused by the potentially non-local accesses for the binary searches of leapfrog-join.

In our approach, we partition the $n$-dimensional search space 
into ``hyper-cubes'' or \emph{boxes} such that the required data for an individual box fits into memory. LFTJ is then run over each box individually---finding all input data ready in memory.
  We strive for the following properties: (i) Determining box-boundaries is efficient: both in CPU and I/O work. (ii) Loading data that is restricted to a box is efficient, again, both in terms of CPU and I/O work. (iii) The total amount of data loaded is minimal. 
\Figref{FigKeyFigure} illustrates this strategy for \lftjtri. 
The join uses three variables $x$, $y$, $z$ -- resulting into a 3-dimensional search-space. 
If the input graph $G$ represented via a TrieArray does not fit into the available memory, then we partition the search space into 
boxes, for example as in \Figref{FigExampleGraphTiling}. The partitioning is chosen such that the input data restricted to an individual box 
fits into memory. \lftjtri is then run for each box individually one after another while join results are written append-only in a streaming fashion.

We now explain the different aspects in detail.

\subsection{TrieArray Slices} 

We assume that input data is given on external storage in a TrieArray representation, with the attribute order consistent with the 
chosen key order for LFTJ. 
This can easily be achieved via a pre-processing step that costs $O(SORT(|I|)$ block I/O and CPU steps.
When loading data for a single box into main memory, we directly operate on the TrieArray representation 
to subset the data. 
The remainder of this subsection shows that this step can be done very efficiently.

In general, applying any selection $\sigma$ 
to a TrieArray for a predicate $R$ to obtain a TrieArray for $\sigma(R)$ can be done in $O(|R|)$ cpu work and
$O(|R|/B)$ I/Os if $\sigma(t)$ can be computed in $O(1)$ time and space for tuples $t \in R$. 
This is because TrieArrays can be used to efficiently enumerate the represented tuples in lexicographically order, 
and they can also efficiently be built from lexicographically sorted tuples.

We are interested in certain \emph{range-based} selections. It turns out that these can be built even faster---with costs
proportional to the selected size $|\sigma(R)|$ rather then the total data set size $|R|$ (modulo log-factor), or even 
less depending on the cost-model. 

\begin{example}[TrieArray Slice] Consider the binary relation $E$ from \Figref{FigExampleGraphADirected} and its TrieArray $T$ 
in \Figref{FigExampleGraphTrieArray}. We are interested in the subset $S$ of $T$ that restricts the first attribute 
to the interval $[3,5]$, i.e., $S = \set{(x,y) \in T \st x \in [3,5]}$. We call this a slice of $S$ at level $0$ from $3$ to $5$. A TrieArray for this slice is shown at the top of \Figref{FigExampleGraphTrieSlice}. 
To build this slice, we can simply copy the values in $\mathbf{val_0}$ for the interval $[3,5]$; then look 
up where the corresponding $y$ values are in $\mathbf{idx_0}$, and copy these as well. The index value cannot 
simply be subset because the positions need to be shifted by the amount we cut off from $\mathbf{val_1}$ in the front:
the first index in $\mathbf{idx_0}$ of the slice should read $0$ instead of $5$. However, instead of changing the 
values, we add a wrapper to the index arrays that can subtract the offset (here 5) during accesses dynamically. 
Then, all data used in the arrays of the slice are simply sub-arrays of the original data. 
\end{example}

In general, for an $n$-ary relation $R$, we are interested in creating slices at a level $k$, $0 \leq k < n$.  At level $k$ the values are restricted to an interval given by a low-bound $l$ and a high-bound $h$; at levels $0,\dots,k-1$, the slice contains only a single element each. 
Formally:
\begin{definition}[Slices]
\label{DefTrieArraySlice}
Let $R$ be an $n$-ary relation, $0 \leq k < n$ an integer, $s$ be a $k$-ary tuple, and $l$ and $s$ be two domain values. The \emph{Slice $S$ of $R$ at level $k$ for $s$ from $l$ to $h$} (in symbols $S = \TrieSlice{R}{s}{l}{h}$) is the defined as:
\[
   S = \set{ (x_0,...,x_{n-1}) \in R \st (x_0,...,x_{k-1}) = s \textnormal{ and } l \leq x_k \leq h }
\]
We often do not mention the level explicitly as it is evident from the start tuple $s$; also, if $k = 0$, we simply say ``Slice for $R$ from $l$ to $h$''. 
\end{definition}

We create and store Slices in the TrieArraySlice data structure, which is a conventional TrieArray---except that the index arrays can be parameterized with an offset to perform dynamic index-adaptation as explained in the example above. As with TrieArrays, we identify the Slice (set of tuples) with the TrieArraySlice data structure and vice versa in the rest of the paper.

Given a relation $R$ on secondary storage, we can create slices of $R$ efficiently:
\begin{proposition}[Slice provisioning] Let $R$ be an $n$-ary relation stored on secondary storage as a TrieArray; 
$k, s, l,$ and $h$ be as in \defref{DefTrieArraySlice}. 
Then, the slice $S = \TrieSlice{R}{s}{l}{h}$ can be loaded into memory with $O(\log |R| + |S|/B)$ block I/Os and $O(\log |R| + |S|/B)$ CPU work, if it fits.
\end{proposition}
\textsc{Sketch} The provisioning process is as follows: using $k$ binary searches on the value arrays $\mathbf{val_0},\dots,\mathbf{val_{k-1}}$, we locate the prefix $s$ in $R$; the slice is empty if the prefix does not exist. Then, using two more binary searches we locate the smallest element $l' \ge l$ and the largest $h' \leq h$ in $\mathbf{val_k}$ of $R$. Their positions are the boundaries in $\mathbf{val_k}$ and $\mathbf{idx_k}$ for the 
interval we copy into the slice. For the remaining $n-k$ value arrays and $n-k-1$ index arrays, we iteratively 
follow the pointers within the $\mathbf{idx}$ arrays and copy the appropriate ranges. As a last step we adjust the index-array's offset parameter: for each $j=k,\dots,n-2$, we set the offset parameter of $\mathbf{idx_j}$ to $-\mathbf{idx_j}[0]$. 

We require $O(\log |R|)$ I/Os for the binary searches and $O(|S|/B)$ I/Os for copying the continuous values from the arrays with indexes $\ge k$. Similarly, the binary searches require $O(\log |R|)$ CPU work; the remaining CPU work accounts for requesting the copy operations. $\qed$

\myparaNN Note that besides the logarithmic component, provisioning a slice amounts to simply copying large, continuous arrays from secondary storage into main memory. On modern hardware, these can be done using DMA methods without causing any significant CPU work. 
Moreover, modern kernels might simply memory map the to-be-copied pages and perform actual copies only when pages are modified.

\mypara{Probing} As the last building block, we are interested in provisioning slices that will fill up a certain budgeted amount of memory. In particular, we specify the prefix-tuple $s$ and lower bound $l$ as before. But instead of providing an upper bound $h$, we give a memory budget $m$ in blocks as shown in \Figref{FigProbeSingle}. We are then interested in a maximal upper bound $h \ge l$ such that the slice at $s$ from $l$ to $h$ requires no more than $m$ blocks of memory. Note that for skewed data, it is possible that the slice \TrieSlice{T}{s}{l}{h} requires more than $m$ blocks of memory, even when $h=l$. Should this case occur, we report via the sentinel value \texttt{SPILL} instead of returning an upper bound $h$. Not surprisingly, probing is also efficient:

\begin{proposition} 
For a TrieArray $T$ on secondary storage, probing the upper bound for a Slice to fill up a memory budget as described in \Figref{FigProbeSingle} requires $O(\log|T|)$ I/Os and CPU work.
\end{proposition}
\textsc{Sketch} Similar to slice provisioning, except that we do a binary search for the upper bound and check for each guess how many blocks the TrieSlice would occupy. This can be done by following the $\mathbf{idx_i}$ pointers. Determining the size of the TrieSlice for each guess requires at most $O(n)$ I/Os where $n$ is the arity of $T$. Since we binary search in $\mathbf{val_k}$, an array that is at most size $|R|$, we obtain the required complexity of $O(\log|R|)$. $\qed$

\begin{figure}[t]
\hrule \vspace{0.5em}
\begin{algorithmic}[0]
\Function{Probe}{$T$, $s$, $l$, $m$} \textbf{returns} $h$
\Statex \textbf{in:} $n$-ary TrieArray $T$ \Comment{on secondary storage}
\Statex \textcolor{white}{\textbf{in:}} \TupleK $s$ \Comment{start tuple for attributes $0,\!..,k-1$}
\Statex \textcolor{white}{\textbf{in:}} value $l$ \Comment{Lower bound for attribute ~$k$}
\Statex \textcolor{white}{\textbf{in:}} \Int $m$ \Comment{memory budget in blocks}
\Statex \textbf{out:} %
 Maximal $h \ge l$ such that the slice \TrieSlice{T}{s}{l}{h} occupies 
\Statex \textcolor{white}{\textbf{out:}} $\leq m$ blocks of memory, or \texttt{SPILL} if no such $h$ exist.
\EndFunction
\end{algorithmic}
\vspace{0.25em}
\hrule
\caption{\label{FigProbeSingle} Interface for Single Slice Probing}
\end{figure}

\subsection{Boxing Procedure}
To help exposition, we first describe aspects of the boxing approach via examples, before we cover
the general case.

\mypara{Joins with one variable} Consider a join over multiple unary relations such as
\[
 Q(x) \leftarrow R(x), S(x), T(x).
\]
Imagine each of the body relations is larger than the available internal memory $M$. 
We can divide the internal memory into four parts, one for the output data and one for each of
the input relations. Since the output is written append-only, a relatively small portion of memory, which is
written to disk once it fills up, is sufficient. We thus divide up the bulk of the memory for the three input 
relations. We can use the simple strategy to \emph{evenly} divide the space. A boxed LFTJ execution 
would then simply alternate probing, provisioning, and calling LFTJ as described in \Figref{AlgExampleSimpleBoxing}.

\begin{figure}
\hrule \vspace{0.5em}
\begin{algorithmic}[1]
\State $l$ \Gets \NegInf \Comment{Value at the start of the search space}
\Repeat
   \State probe $R,S,T$ from $l$ for upper bounds $h_R$, $h_S$, $h_T$ 
   \State $h \Gets \min(h_R, h_S, h_T)$
   \State provision $R,S,T$ from $l$ to $h$
   \State run LFTJ on the provisioned slices
   \State $l \Gets \Successor{h}$ \Comment{lower bound is successor of old upper}
\Until{$\Inf = h$} \Comment{until we have searched all space}
\end{algorithmic}
\vspace{0.25em}
\hrule
\caption{\label{AlgExampleSimpleBoxing}Example: Boxing for $R(x),S(x),T(x)$}
\end{figure}

Not surprisingly, this approach would work well for the limited class of joins: for reading the input, it requires a number of I/Os 
bound by $O(|I|/B + |I|/M\log|I|)$ with $|I|$ being the combined size of the input relations. The key 
observations for showing the bound is that in each iteration (except possibly the last), 
at least one relation will load $O(M)$ (in our example around  $M/3$) tuples using $O(M/B)$ block reads. Now,
since there are only $|I|$ tuples in the input, there are at most $O(|I|/M)$ iterations. Since each
probing can be done in $O(\log|I|)$ we obtain the desired bound\footnote{Note that with a simple caching strategy for the provision step (always cache the block containing $h$ and reuse in the next provisioning if possible), we could 
make the argument that each block is read at most once by the provisioning step obtaining the same asymptotic bound.}.

\begin{figure}
\hrule \vspace{0.5em}
\begin{algorithmic}[1]
\Statex \textbf{variables:} \TupleM \gLow, \gHigh \Comment{Box boundaries}
\Procedure{Main}{}
   \State \Call{BoxUp}{1}
\EndProcedure
\Procedure{BoxUp}{\Int \Var} \Comment{\Var corresponds to $x_\Var$}
   \State \gLowAt{\Var} \Gets \NegInf 
\Repeat
   \State probe inputs $R_i$ from \gLowAt{\Var} for upper bound $h_i$
   \State $\gHighAt{\Var} \Gets h_i$
   \State provision $R_i$ from $\gLowAt{\Var}$ to $\gHighAt{\Var}$

   \If{$\Var < m$} :
      \Call{BoxUp}{i+1}
   \Else :
      run LFTJ on slices \Comment{Box: $[\gLow\!\cdot\!\cdot\!\cdot\!\gHigh]$}
   \EndIf
   \State $\gLowAt{\Var} \Gets \Successor{\gHighAt{\Var}}$ 
\Until{$\Inf = \gHighAt{\Var}$} 
\EndProcedure
\end{algorithmic}
\vspace{0.25em}
\hrule
\caption{\label{AlgExampleSimpleBoxingCrossProduct}Example: Boxing for $R_1(x_1),\!...,R_m(x_m)$}
\end{figure}

\mypara{Unary cross-products} Consider the cross-product of $m$ unary relations, with each relation larger than $M$:
\[
 Q(x_1,\!...,x_m) \leftarrow R_1(x_1), \dots, R_m(x_m).
\]
We again split the bulk of the available memory across the $m$ input relations. The boxing procedure is recursive where each \emph{dimension} $\Var$ of the recursion corresponds to a variable $x_i$ (See \Figref{AlgExampleSimpleBoxingCrossProduct}). The procedure starts with $\Var = 1$. In general, at a dimension $i$, we loop over the predicate $R_i$ via the probe-provisioning loop. Then, for each slice at dimension $i$, we do the same recursively for the next higher dimension. At the bottom of the recursion---when we reached the $i=m$, we call LFTJ on the created slices. 
Then, the slices provide data for the box $[\gLow\!\cdot\!\cdot\!\cdot\!\gHigh]$, i.e., in which the variable $x_i$ can range from $\gLowAt{\Var}$ to $\gHighAt{\Var}$. Note that (like above) we can run the original query over the slice data since the slices are guaranteed to not have data outside their range and thus the boxes partition the search-space without overlap.

\mypara{General joins} The general approach combines the two previous algorithms while also considering corner cases. Let $Q$ be a general full-conjunctive join of $m$ atoms, and variable order $\pi = x_1,\dots,x_n$ with no atom containing the same variable twice, and all atoms in $Q$ mentioning variables consistent with $\pi$. We first group the atoms based on their first variable $x_j$: we place all atoms that have as first variable $x_j$ into the array \Preds{1..n} at position $j$.  To follow the exposition, consider the join
\[
  Q(x_1,x_2,x_3) \leftarrow R(x_1,x_2), S(x_1,x_3), T(x_2,x_3), U(x_1)
\]
where we put $R,S$, and $U$ into \Preds{1} and $T$ into \Preds{2}. Like for cross-products, we recursively provision for the dimension $i$ ranging from $1$ to $n$. For each $i$, we use the method for joining unary relations for the atoms in \PredsT{i}. In particular, for each  $A_j \in \Preds{i}$ we probe and create slices for $A_j$ \emph{at level 0} regardless of $i$ or the arity of $A_j$.  
Thus, at dimension $i$, we iteratively provision atoms with $x_i$ as their first attribute restricting the range of $x_i$ but not any of the other variables $x_k$, $k>i$. This ensures that we can freely choose any partitions we might perform on these variables $x_k$ for $k>i$. Like with cross-products, we call LFTJ at the lowest level when $i=n$. 

The above works well unless any of the probes reports a \texttt{SPILL}, which can occur if a relation exhibits significant skew. For example, imagine there is a value $a$ for which $|S_a(x_3)|$ exceeds the allocated storage. Then, at dimension $i=1$, 
probing $S$ at level $0$ with a lower bound $a$ will return \texttt{SPILL}. We handle these situations by setting the upper bound at level $i=1$ to $a$, and essentially marking $S_a$ as a relation that needs to be provisioned at the dimension of its second attribute (eg, 3) alongside the atoms in \PredsT{3}. Note that a relation of arity $\alpha$ can spill $\alpha - 1$ times in worst case.

The general algorithm is given in \Algref{AlgTileGeneral}. We evenly divide the available storage among the $n$ dimensions, and assign the atoms $A$ to \PredsT{\Var} accordingly (lines 3-4). We also use a variable \LeftOverMemT to let lower dimensions utilize memory that was not fully used by higher dimensions. In line 11, we union the spills from the previous level to the atoms we need to provision. The method \texttt{probe} in line 12, probes atoms in \texttt{atms} to find an upper bound such that all atoms can be provisioned. We here, evenly divide \texttt{mem} by the size of \texttt{atms}.
The lower bound for probing are taken from \texttt{low}, which is also used to determine 
the starting tuples for possible spills. The method sets the upper bound at the current dimension and fills the spills predicate 
if necessary. The method \texttt{provision} provisions the predicate $A$ with bounds from \texttt{low} and \texttt{high} adapted to the variables occuring in $A$. 
It returns the slice and the size of used memory.

\begin{algorithm}[t]
\caption{\label{AlgTileGeneral}Boxing Leapfrog Triejoin}%
\begin{algorithmic}[1]
\Statex \textbf{in:} \maxMem \Comment available memory in blocks
\Statex \textcolor{white}{\textbf{in:}} $A_1,\dots,A_m$ \Comment body atoms and TrieArrays
\Statex \textcolor{white}{\textbf{in:}} $x_1,\dots,x_n$ \Comment key order, $n$ variables
\Statex \textbf{variables:} %
\Statex \textcolor{white}{\textbf{in:}} \Tuple \gLow, \gHigh \Comment Box boundaries
\Statex \textcolor{white}{\textbf{in:}} \ArrayOf \AtomSet  \Preds{1..n}  \Comment{atoms per level} %
\Statex \textcolor{white}{\textbf{in:}} \ArrayOf \SliceSet \Slices{1..n} \Comment{provisioned slices}%
\Statex \textcolor{white}{\textbf{in:}} \ArrayOf \AtomSet  \Spills{0..n} \Comment{spilled-over atoms}%
\Statex \textcolor{white}{\textbf{in:}} \ArrayOf \Int \BudgetVar{1..n} \Comment of memory in blocks
\Statex 
\Procedure{Main}{}
\For{$\Var \in \set{1,\dots,n}$}
    \State \BudgetVar{\Var} $\gets$ $\maxMem / n$ %
    \State \Preds{\Var} $\gets \set{ A_\Var \st x_\Var \textnormal{ is first variable in } A_\Var }$
\EndFor
\State \Call{\BoxLevel}{1, 0} \Comment{1st variable, no leftover memory}
\EndProcedure
\Statex
\Procedure{\BoxLevel}{\Var, \LeftOverMem}
\State \Budget \Gets $ \BudgetVar{\Var} + \LeftOverMem$
\State \gLow \Gets \NegInfTuple ; \gHigh \Gets \InfTuple
\Repeat
    \State \Slices{\Var} \Gets $\emptyset$ ; \UsedMemAccu \Gets 0
    \State \LocalAtoms \Gets \Preds{\Var} $\cup$ \Spills{\Var - 1}
    \State \Spills{\Var}, \gHighAt{\Var} \Gets probe(\Var, \LocalAtoms, \Budget, \gLow)
    \For{$A \in \LocalAtoms \setminus \Spills{\Var}$}
        \State \LocalSlice, \LocalMem \Gets provision($A$, \gLow, \gHigh)
        \State \UsedMemAccu \Gets \UsedMemAccu + \LocalMem
        \State \Slices{\Var} \Gets \Slices{\Var} $\cup$ \LocalSlice
    \EndFor
    \If{$\Var < n$}
       \State \LeftOverMem \Gets $\Budget - \UsedMemAccu$
       \State \Call{\BoxLevel}{\Var + 1, \LeftOverMem}
    \Else
       \State \textbf{run} LFTJ on $\bigcup\limits_{k=1..n} \!\!\!\!\Slices{k}$ on Box[$\gLow\!\cdot\!\cdot\!\cdot\!\gHigh$]
    \EndIf 
    \State \gLowAt{\Var} \Gets \Successor{\gHighAt{\Var}}
\Until \Inf = \gHighAt{\Var} 
\EndProcedure
\end{algorithmic}
\end{algorithm}

\subsection{I/O Complexity of Boxing}

We now analyze the Boxing approach to obtain complexity bounds on the number of block I/Os. Since 
we concentrate on full conjunctive queries, every output tuple is computed exactly once by LFTJ. As explained above,
we use some constant-size buffer to let the I/O cost for the output be $K/B$ where $K$ is the output size. 
We now analyze the cost of the I/Os for reading input data.

For each dimension $i$, $i=1,\!...,n$, let \LoopNums{i} be an upper bound on how often the repeat-until loop from 
lines 9--23 of \Algref{AlgTileGeneral} is executed for a single invocation of the surrounding \Call{\BoxLevel}{} procedure. 
\LoopNums{i} is determined by how often we need to provision to completely iterate through the atoms 
in $\PredsT{\Var} \cup \Spills{\Var}$.
In each step (except possibly the last) at least one of the input predicates $A_j$ loads $O(M)$ tuples---this is 
the predicate that determines the high bound \gHighAt{\Var}. In case there is no spill, this is immediately clear; but it is even true if 
a predicate is being spilled because its tuples are then ``consumed'' at a higher dimension. Note, that at the last dimension,
no spills can occur.
We thus have $L_i \in O(m|I|/M) = O(|I|/M)$, where $m$ is the number of atoms in the join.

Let us now determine how often for each dimension \Call{\BoxLevel}{} is called. We denote this number by \BoxNums{i}.
The outermost \Call{\BoxLevel}{}
is called once; \Call{\BoxLevel}{2,~.} is called once for each iteration of the repeat 
loop at level 1, that is $\LoopNums{1}$ times. In general, $\BoxNums{i} = \prod_{j=1}^{i-1} \LoopNums{j}$ and consequently 
$\BoxNums{i} \in O(|I|^{i-1}/M^{i-1})$. 

It is convenient to inject the following observation: %
\begin{lemma} \label{LemmaNumboxes} The number of boxes created by a boxed LFTJ with $n$ variables on input $I$ is $O(|I|^n/M^n)$. In particular, if $|I|\in O(M)$ then the number of boxes is $O(1)$.
\end{lemma}
\textsc{Proof} The number of boxes equals the number of loop executions at dimension $n$, which is bound by $\LoopNums{n}\BoxNums{n}$. $\qed$

\myparaNN
Back to the I/O costs. Consider only the I/O that is performed
directly in a certain \Call{\BoxLevel}{} call without counting the cost in the recursive 
calls from line 19. First, we count provisioning only. Here, during the evaluation of the repeat 
loop (lines 9-23), we load the data in $\PredsT{.} \cup \SpillsT{.} \subseteq I$. %
Similarly as in the case of joins with one variable, we can cache the last blocks containing the last tuple of the provisioned TrieSlices, and thus load each block from the input exactly once.
Consequently, the I/O work done to provision directly in each invocation of \Call{\BoxLevel}{i} is limited by $O(|I| / B)$.
The I/Os necessary for probing can be bound by $O(\LoopNums{i}\log|I|) = O(|I|/M\log|I|)$ since we probe 
at most $m$ relations once for each execution of the repeat loop. If we use the assumption that
$|I|/B$ is larger than $\log{|I|}$ as explained in \Secref{SecAssumptions}, we thus obtain $O(|I|/B)$
as I/O cost directly at dimension $i$ for a single \Call{BoxLevel}{} call. As last step, we 
multiply by $\BoxNums{i}$ to obtain the total I/O cost $C_i$ at dimension $i$ as $C_i \in O(|I|^i/(M^{i-1}B))$.  
Since output is written once and we consider joins without projections we obtain:
\begin{theorem} \label{thmGeneralIO}
The I/O complexity of boxed LFTJ with $n$ variables, input $I$ and output of size $K$ is $O(|I|^n/(M^{n-1}B) + K/B)$.
\end{theorem}

\mypara{Optimality} 
This complexity is optimal when only the number $n$ of variables is used to characterize the query. This is because the Cartesian product of $n$ relations can produce $\Theta(|I|^n)$ output which requires $\Theta(|I|^n/B)$ block writes.

Furthermore, in practice, the input is often only by a constant factor larger than the available memory: 
\begin{corollary} \label{CorBestComplexity} The I/O complexity of boxed LFTJ for any query on input $I$ and output of size $K$ is $O(|I|/B + K/B)$ if $I \in O(M)$. 
\end{corollary} 
This (better looking) bound is, obviously, optimal for queries that require reading the entire input. 

\mypara{No spills} If the execution does not produce any spills, we can strengthen the general result. To do so, 
we quickly need to introduce a property of queries:
\begin{definition} 
The \emph{rank $r_\pi(Q)$ of a query $Q$ conforming to the key-order $\pi=x_1,...,x_n$} is the largest $j$ 
for which $Q$ contains an atom with $x_j$ as first variable. The \emph{rank $r(Q)$ of $Q$} is the 
minimum of $r_\pi(Q)$ where $\pi$ is any key-order.
\end{definition}
Clearly, the rank of a query (for any key-order) is bound by the number of variables---but sometimes smaller.
E.g., for the triangle query $r_{x,y,z}\eqref{EQtriangleQuery} = 2$, but also $r\eqref{EQtriangleQuery}$ is 2. 
Note that $r_\pi(Q)$ is the largest $i$ for which \PredsT{\Var} is non-empty when boxing $Q$ with key-order 
$\pi$. 

\begin{theorem} \label{thmSpillIO} If no spills occur during a boxed execution of LFTJ for the query $Q$ with key-order $\pi$, then the total I/O cost is 
$O(|I|^\ell/(M^{\ell-1}B) +K/B)$ where $|I|$ denotes the combined size of the input relations, $K$ the combined size of the output relations and $\ell = r_\pi(Q)$ is the rank of $Q$ for $\pi$.
\end{theorem}
\textsc{proof} At dimensions $i > \ell$, there are no I/O operations since both \PredsT{i} and \SpillsT{i} are empty, obtaining the desired 
result by summing up $C_i$ for $i \leq \ell$. $\qed$

\myparaNN  Spills occur in the boxed LFTJ execution if there is an input relation $R$ and value $a$ for which the 
Slice $R_{a\rightarrow a}$ exceeds the size of the memory allocated for $R$.
We can thus characterize when they occur. For a query with $n$ variables and $m$ atoms: Let $M'$ be the memory used for the body of 
the query. If we divide up all space evenly among all $n$ variables, and for each dimension, evenly among all $m$ predicates, 
then the critical value for any 
$|R_a|$ is approximately $M'/(2nm(k-1))$, since the slice for $R_{a \rightarrow a}$ has a size of at most around $2(k-1)|R_a|$.
\TODO{Possibly see appendix for exact values}

\subsection{CPU Data-Complexity of Boxing}

The CPU work performed by a boxed \lftj on input data $I$ falls into two categories: (1) the work necessary to 
determine the number of boxes and to provision them, and (2) the work done by the in-memory \lftj executing over the 
boxes. For an input $I$, the asymptotic work in category 2 is trivially bound by
the asymptotic work of the in-memory \lftj on $I$ multiplied by the number $P$ of boxes used, simply because each 
invocation uses input that is a subset of $I$. For the work in category 1: deciding on the bounds of a single box is done in $O(\log|I|)$,
copying its required data takes no more than $O(|I|)$ resulting into a total upper bound of $O(P|I|)$ for $P$ boxes.

Using \lemmaref{LemmaNumboxes}, 
we can thus conclude:
\begin{theorem}  \label{ThmCPUComplexityDoesntChange}
On inputs $I$ that are only by a constant factor larger than the available memory $M$,
the asymptotic computational data-complexity of the boxed LFTJ matches the one of the in-memory version of LFTJ
or is linear in $|I|$, whichever is worse.
\end{theorem}

\section{LFTJ Applied to Triangle Listing} \label{SecLFTJTriangle}

\subsection{Boxed LFTJ-$\mathbf{\Delta}$} 
From Corollary~\ref{CorBestComplexity}, we immediately get 
an I/O complexity of $O(|E|/B+K/B)$ if $|E|\in O(M)$.
Without this assumption, plugging the triangle query $\eqref{EQtriangleQuery}$ 
into Theorems \ref{thmGeneralIO} and \ref{thmSpillIO}, we obtain:
\begin{corollary}
The boxed LFTJ applied to the triangle query has an I/O complexity of $O(|E|^3/(M^2B)+K/B)$. 
If there are no spills the complexity is $O(|E|^2/(M B) + K/B)$.  
\end{corollary}
With no spills, boxed LFTJ thus matches the I/O complexity of MGT \cite{hu2013massive}, which is 
optimal if $M \ge |V|$ as shown in \cite{hu2013massive}. From above, we know that spills only occur if there
is a single node that has more than around $M/18$ neighbors, for 5GiB of allocated memory and 64bit node ids, 
this amounts to an upper limit of 37 million neighbors per node, a number that is seldom reached in practice. Interestingly, the \emph{core} MGT algorithm in \cite{hu2013massive} also requires that the node degree is limited. MGT achieves the bound without restrictions 
by deploying a pre-processing step.

For the compute complexity of boxed \lftjtri, we rely on \theoremref{ThmCPUComplexityDoesntChange},
expecting to essentially maintain the performance of in-memory \lftjtri assuming $|I|\in O(M)$. 

\subsection{In-Memory LFTJ-$\mathbf{\Delta}$ CPU Complexity} \label{SecTriangleCPUComplexity}

In this section, we use the conventions that $G=(V,E)$ is always the input graph. %
While the previous section was specific to our version of LFTJ that uses TrieArrays, the results here apply
to all LFTJ implementations as long as the basic TrieIterator operations adhere to the complexity bounds
given in \cite{Veldhuizen14} and restated in \Secref{SecLFTJcomplexityrequirement}. 
Following with little work directly from \cite{Veldhuizen14} and \cite{atserias2008size}: %
\begin{proposition} LFTJ-$\Delta$'s computational complexity on input $E$ is $O(|E|^{1.5}\log |E|)$, which is optimal modulo the log-factor.
\label{propSimpleCPUComplexity}
\end{proposition} 
\begin{proof} See \appref{apppropSimpleCPUComplexity}. \end{proof}

The rest of the section, strengthens this result by analyzing the complexity of LFTJ on families of graphs that are characterized by the number of edges and their \emph{arboricity}. The arboricity $\alpha(G)$ of an undirected graph $G$ is a standard measure for graphs, counting the minimum number of edge-disjoint forests that are needed to cover the graph. A classic result by Nash-Williams \cite{NashWilliams01011964} links this number to the graph's density by showing that no subgraph $H$ of $G$ has more than $k(|V(H)|-1)$ edges if and only if $\alpha(G) \leq k$.
In general, $\alpha$ is in  $O(\sqrt{|E|})$ \cite{chiba1985arboricity} for any graph $G=(V,E)$. However, in many real-world graphs, $\alpha$ is significantly smaller\cite{chiba1985arboricity,lin2012arboricity}.

It turns out that the runtime-complexity for \lftjtri is related to the graph's arboricity with \lftjtri behaving better the smaller $\alpha$ is. It thus makes sense to consider \lftjtri's complexity for graphs characterized by an upper bound on their arboricity. For compatibility with the asymptotic complexity, we bound the graph's arboricity with respect to their edge-size:
\begin{theorem} 
\label{ThmLftjcomplexityalpham}
Let $\hat\alpha : \mathbb N \rightarrow \mathbb N$ be a monotonically increasing
function. Then, \lftjtri runs in $O(m\hat\alpha(m)\log m )$ time on graphs with at most $m$ edges and arboricity of at most $\hat\alpha(m)$.
\end{theorem}
\begin{proof}
See \appref{AppLftjalphabound}. Analyze the work done by the leapfrog joins at levels $x$, $y$, and $z$. Only the third level is interesting, where we use a result by \cite{chiba1985arboricity} that gives an upper bound of $2\alpha(G)|E|$ for the sum $\sum_{(x,y)\in E} \min\set{d(x),d(y)}$.
\end{proof}

Clearly, if the maximum degree of our graphs is bounded, than their arboricity is in $O(1)$. Furthermore,
the arboricity of planar graphs is also in $O(1)$
\cite{chiba1985arboricity}, immediately leading to:
\begin{corollary}  \lftjtri lists triangles in $O( |E| \log |E|)$ steps for planar graphs and for graphs with bounded degree.
\end{corollary}

We can also amend the optimality result from \propref{propSimpleCPUComplexity} showing
that \lftjtri remains optimal (modulo log-factor) even when considering graphs with a limited arboricity:
\begin{theorem} \label{ThmFineGrainedOptimal}
Let $\hat\alpha : \mathbb N \rightarrow \mathbb N^+$ be a monotonically increasing, computable function that is not identical to $1$ and in $o(\sqrt{n})$. Then, no algorithm that lists all triangles for input graphs $G=(V,E)$ \emph{with} arboricity of at most $\hat\alpha(|E|)$ can run in $o(|E|\hat\alpha(|E|))$ time.
\end{theorem}
\begin{proof} See \appref{AppFineGrainedOptimal}. It turns out that for any such $\hat\alpha$, we can construct large graphs that have $\Theta(|E|\hat\alpha(|E|))$ triangles. \end{proof}

We highlight that the above theorem is quite general. It only requires the algorithm to be correct for input graphs of restricted arboricity\footnote{Except for the corner-case where the arboricity is bound by 1, in which case the graphs have no triangles and an $O(1)$ algorithm trivially exists.}. For example, even if we (somehow) knew that all our input graphs have an arboricity $\alpha$ bound by, say, $42\log|E|$, we could not design a specialized algorithm that only works on these graphs \emph{and} has a runtime complexity of $o(|E|\log|E|)$. 

\myparaNN The optimality from \theoremref{ThmFineGrainedOptimal} does unfortunately not directly follow from the 
worst-case optimality of LFTJ for families of instances that are closed under renumbering (Thm~4.2 in\cite{Veldhuizen14}), 
because the optimality in \cite{Veldhuizen14} was obtained when each relation symbol appears only once in the body of the join, a property used in the proof for Thm~4.2 of \cite{Veldhuizen14}.

\section{Implementation} \label{SecImpl}

We have implemented a general-purpose join-processing system with LFTJ at its core. To highlight its generality, we briefly list its current features. We support multiple fixed-size primitive data types including int64, double, boolean, and a fixed-point decimal type. Predicates (stored as TrieArrays) can have variable arities and we support marking a prefix of the attributes as key (the TrieArray then needs fewer index arrays). Predicates support loading and storing from and into CSV files. Besides materialized predicates that store data, we have TrieIterator implementations for various ``builtins'' such as comparison operators and arithmetic operators. Using a simple command-shell, joins such as the triangle query can be issued in a Datalog-like syntax. We require the written joins to have atoms with variables consistent with a global key order. At the head of rules, we support optional projections, and some aggregations. The system uses secondary storage (via memory-mapped files) to allow processing of data that exceeds the physical memory; and deploys the here presented Boxing technique. We have not implemented a query optimizer (to find good key orders), nor do we currently support mutating relations, also we do not support transactions. In the following, we highlight aspects of the system that likely have an impact on performance, yet whose detailed analysis and description goes beyond the scope of this paper.

\mypara{Removing interpretation overhead} Datalog queries that are issued are compiled to optimized machine-code and loaded as a shared library into the system. Our code still uses the TrieIterator interfaces but most code is templatized: predicates by their arity, key-length and types of the attribute; TrieIterators by their types and arity; the LFTJ by the key-order, TrieIterators of body atoms as well as each of their variables; a rule by the LFTJ for processing the body and the classes that perform so-called head-actions. Using this approach, we can still program with the convenient TrieIterator interfaces---yet allow the C++ compiler to potentially inline join processing all the way down to the binary searches using the appropriate comparison operators for the type at hand.

\mypara{Misc Optimizations} We are also deploying a parallelization scheme for LFTJ to utilize multiple cores. In the boxed LFTJ version, boxes are worked on one after another, yet LFTJ utilizes available cores while processing a single box. We will also provide single-threaded performance when comparing with single-threaded competitors.

Even though dividing the available memory evenly across the dimensions is sufficient to obtain the asymptotic complexity bounds, using more memory at smaller dimensions reduces the number of boxes created. Note that as long as the memory used at each dimension is a constant fraction of the total memory, the complexity bounds remain in tact. We picked a ratio of 4:1 for dividing up the memory between $x$:$y$ in the triangle query. We also do not allocate budget to dimensions $j$ that do not have an atom using $x_j$ as first variable (eg, $z$). This is fine since in case there is a spill the budget for the spilling relation will be moved over to the next dimension.

If there are two atoms referring to the same relation \emph{and} having the same first variable, we naturally only provision and create one slice for them. For example in the triangle query, we probe and provision a single relation $E$ at dimension $x$ for the atom $E(x,y)$ and the atom $E(x,z)$. Of course, in the case of spills they might get untangled at higher dimensions. We do \emph{not} exploit the fact that the third atom $E(y,z)$ refers to the same relation. 

We envision that for some queries, an optimizer, aided by constraints provided by the user, can avoid provisioning certain boxes because it can infer that there cannot possibly be a query result within that box. For example, in our case, we know that $x<y<z$. This can easily be inferred from the constraint $a<b$ for any $(a,b) \in E$. Based on this, we do not need to provision at dimension $y$ if the high bound for $y$ is smaller than the low bound for $x$. We have put a hook into the boxing mechanism to bypass provisioning if \emph{after} probing this condition is met. A detailed exploration of constraints and their interactions with probing and provisioning is beyond the scope of this work.

\section{Experimental Evaluation} \label{SecExperiments}

In our experimental evaluation, we focus on the triangle listing problem. %
Here, we investigate the following questions: (1) What is the CPU overhead introduced by boxing LFTJ? (2) How well does boxed LFTJ cope with limited available main memory, how does vanilla LFTJ do? (3) How does LFTJ compare to best-in-class competitors?

\mypara{Evaluation environment} We use a desktop machine with an Intel i7-4771 core, that has 4 cores (8 hyper-threaded), each clocked at 3.5GHz. The machine has 32GB of physical memory and a single SSD disk. It is running Ubuntu 14.04.1 with a stock 3.13 Linux kernel.

\begin{figure}
\hspace*{-3ex}\resizebox*{1.09\columnwidth}{!}{
\begin{tabular}{r||r|r|r|r|r|r|r} 
                   & {LJ}        & {Orkut}          & {RAND16}         & {RMAT16}   & {RAND80}                & {RMAT80}       & {Twitter}       \\
 \hline                                                                                                                                      
 \hline                                                                                                                                      
 \vspace*{-2ex}                                                                                                                              
 \\                                                                                                                                          
 csv               & 500MB       & 1.8GB            & 4.1GB            & 4.0GB             & 22GB             & 22GB           &  25GB           \\
 TA                & 315MB       & 1.2GB            & 2.3GB            & 2.2GB             & 11GB             & 11.2GB         &  10GB           \\
 $|V|$             & ~4 Mio      & ~3 Mio           & 16 Mio           & 16 Mio            & 80 Mio           & 80 Mio         &  ~42 Mio        \\
 $|E|$             & 35 Mio      & 117 Mio          & 256 Mio          & 256 Mio           & ~1.28 Bill       & ~1.28 Bill     &  1.2 Bill       \\   
 $\frac{|E|}{|V|}$ & 8.7         & 38.1             & 16               & 16                &  16              & 16             &  28.9           \\
 \#$\Delta$        & 178 Mio     & 628 Mio          & 5457             & 2.2 Mio           & 5491             & 884,555        &  35 Bill        
\end{tabular}}
\vspace*{-2ex}
\caption{Characteristics of the used data sets. \label{FigTableData}}
\end{figure}
\mypara{Data} We use both real-world and synthetic input data of varying sizes. The data statistics are shown in \Figref{FigTableData}. The smallest dataset we consider is ``LJ'', which contains the friend-ship graph of the on-line blogging community LiveJournal \cite{snapnets, YangLJ}. Next, ``Orkut'' is the friend-ship graph of the free online community Orkut \cite{mislove2007socialnetworks, snapnets}. `TWITTER' is one of the largest freely available graph data sets. It contains the as-of-2010 ``follower'' relationships among 42 Million twitter users \cite{Kwak10www}. The dataset has 1.47 billion of these relations, which we interpret as undirected edges in a graph, resulting in 1.2 billion edges. This dataset contains almost 35 billion triangles. Unlike the first two data sets, which we obtained from \cite{snapnets}, twitter was gathered from \cite{twitterUrl}. We also consider synthetically generated data due to its better understood characteristics. We focus on two 
datasets: `RAND' and `RMAT'. Each comes in a medium-sized version with 16 million nodes and 256 million edges and a large version with 80 million nodes and 1.28 billion edges. In the `RAND' dataset, we create edges by uniformly randomly selecting two endpoints from the graph's nodes. The `RMAT' data contains graphs created by the Recursice Matrix approach as proposed by Chakrabarti et al.\cite{chakrabarti2004r}. This approach creates graphs that closely match real-world graphs such as computer networks, or web graphs. We used the data generator available at \cite{RMATGeneratorUrl} with its default parameters. The LiveJournal and the synthetic graphs were also used by the MGT work in \cite{hu2013massive} and earlier work \cite{chu2011triangle} to evaluate out-of-core performance for the triangle listing problem. All graphs have been made simple by removing self and duplicate edges. The CSV sizes in \Figref{FigTableData} refer to the CSV data where each undirected edge $\set{a,b}$ is mentioned only once. TA stands for
our TrieArray representation as described in the earlier sections. We use 64 bit integers per node identifier.

\mypara{Methodology} We measure and present the time for running the algorithms on the mentioned data sets with various configurations and memory restrictions. We will run our TrieArray-based implementation of LFTJ with various configurations and two competing algorithms. Since all algorithms need to report the same number of triangles, we essentially run them in ``counting-mode'' and we thus do not account for the time nor the I/O it takes to output the triangles. This was also done in \cite{hu2013massive}. Input data for LFTJ is given in TrieArray format; we do not
include the time it takes to create the TrieArray from CSV data (which can be done in at most two passes after sorting the data).

\begin{figure*}
\includegraphics[width=0.195\textwidth]{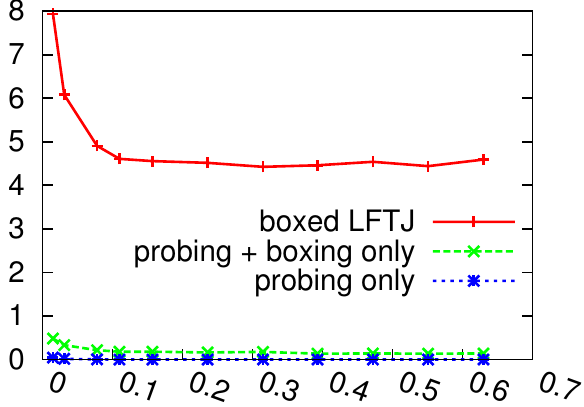}
\includegraphics[width=0.195\textwidth]{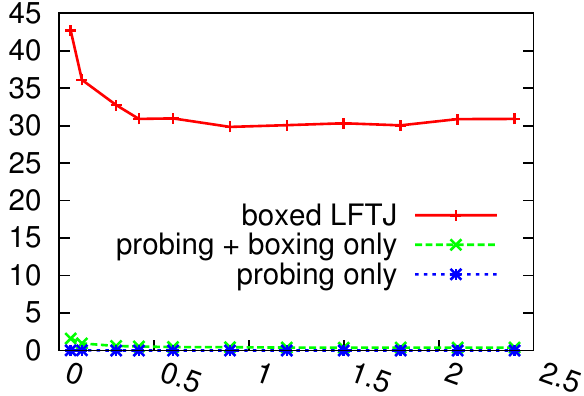}
\includegraphics[width=0.195\textwidth]{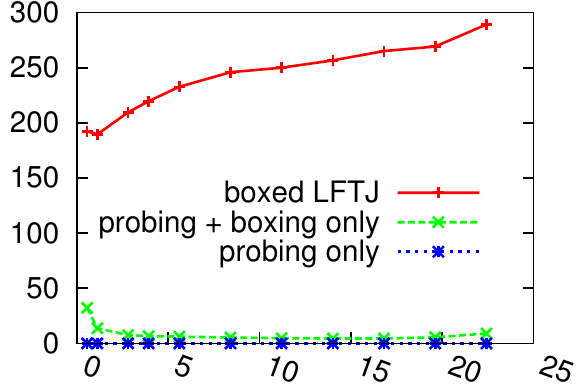}
\includegraphics[width=0.195\textwidth]{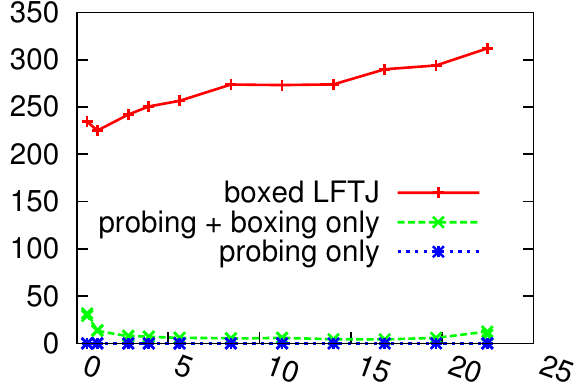} 
\includegraphics[width=0.195\textwidth]{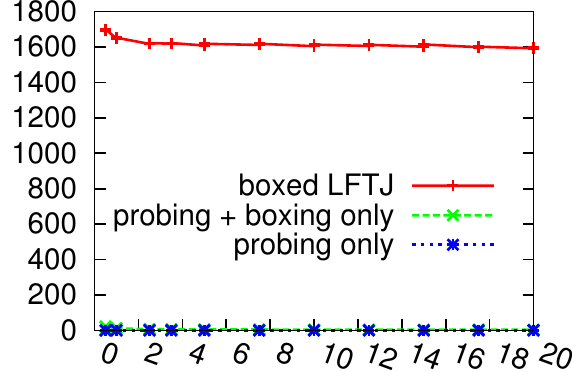}\\
\includegraphics[width=0.195\textwidth]{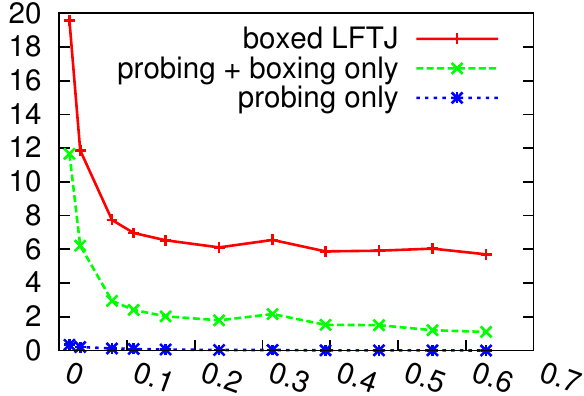}
\includegraphics[width=0.195\textwidth]{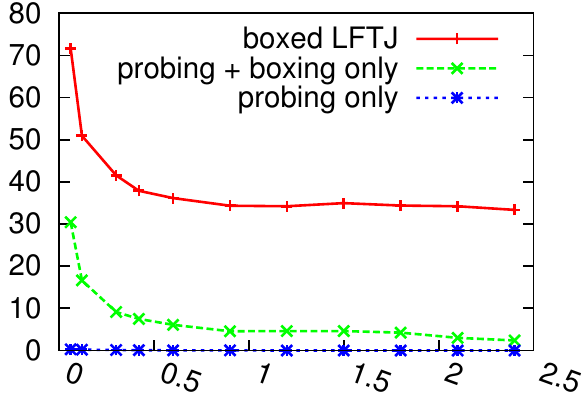}
\includegraphics[width=0.195\textwidth]{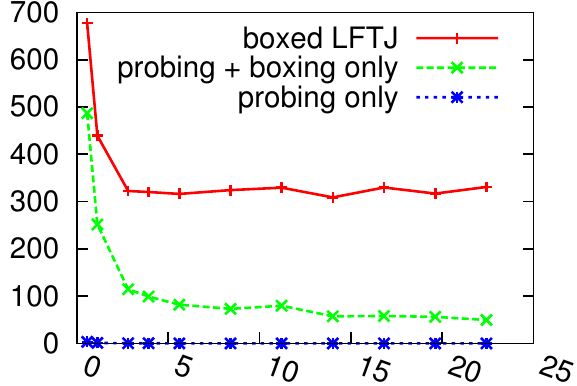}
\includegraphics[width=0.195\textwidth]{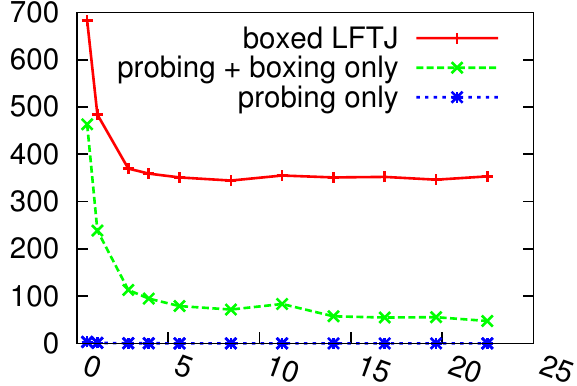} 
\includegraphics[width=0.195\textwidth]{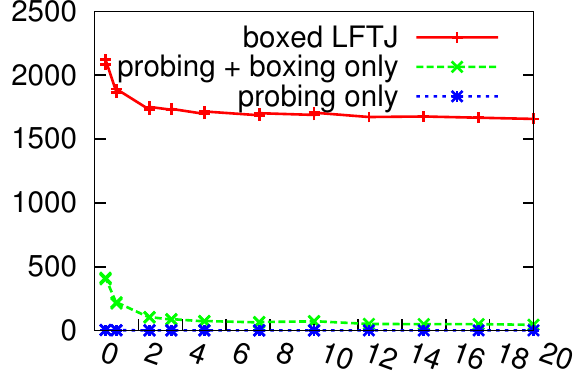}\\
\subfigure[LJ]     {\includegraphics[width=0.195\textwidth]{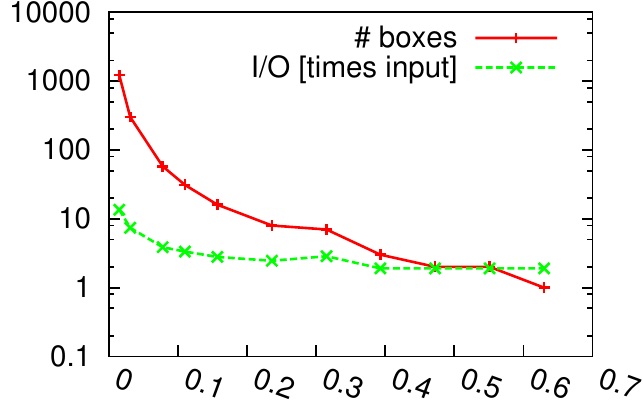}}
\subfigure[ORKUT]  {\includegraphics[width=0.195\textwidth]{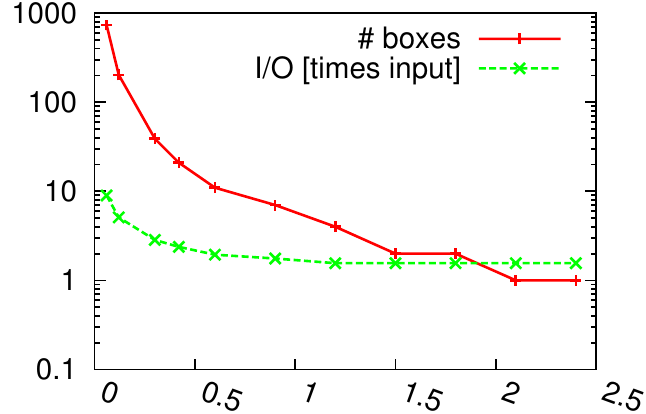}}
\subfigure[RAND80]{\includegraphics[width=0.195\textwidth]{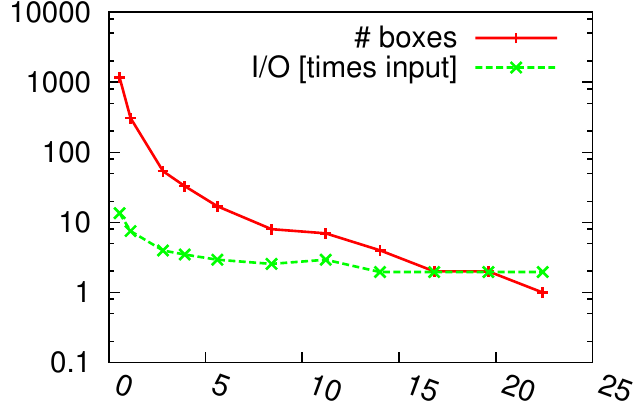}}
\subfigure[RMAT80]{\includegraphics[width=0.195\textwidth]{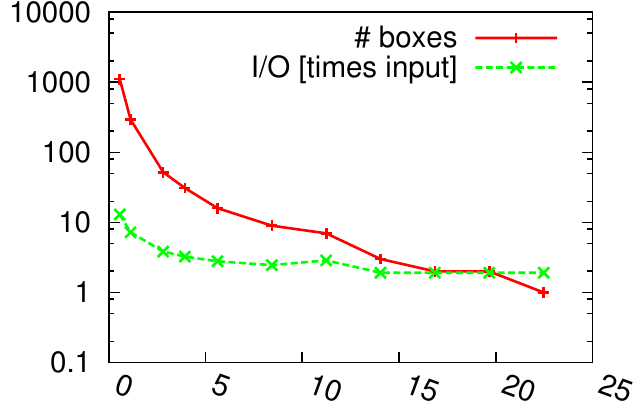}}
\subfigure[TWITTER]{\includegraphics[width=0.195\textwidth]{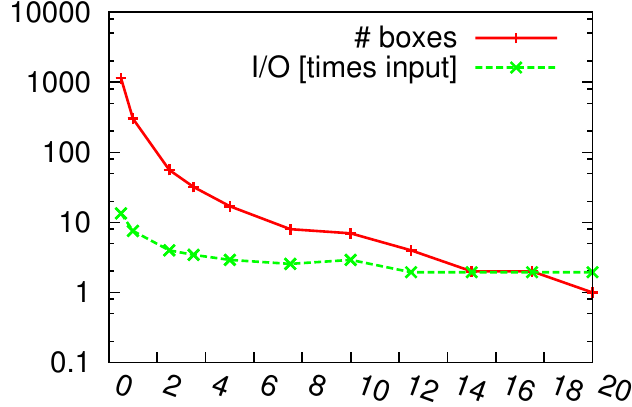}}
\vspace*{-4ex}
\caption{\label{FigBoxingEval}Boxed LFTJ Analysis. \textnormal{On the X-Axis, we vary the amount of total memory available for boxing shown in GB. First row show total runtime in seconds without OS-level memory-restrictions and warm caches to evaluate the additional CPU work necessary for boxing. For performance in an out-of-core scenario, we enforce OS-level memory restrictions and have all caches cleared before execution in the second row. The third row shows the number of boxes and the amount of provisioned memory in multiples of the size of the input data. Omitted graphs for \{RAND|RMAT\}16 look like the ``80'' variants.}}
\vspace*{-2ex}
\end{figure*}

\myparaND{What CPU overhead does Boxing introduce?} To measure the CPU overhead that is introduced by the boxing approach, we advise LFTJ to only use memory the size of a fraction of the input during execution---yet, we do not place any limit on the caches the operating system keeps for file operations. To further (almost completely) remove I/O, we prefix the execution by \texttt{cat}-ting all input data to \texttt{/dev/null}, which essentially pre-loads the Linux file-system cache. We now consider the two questions (i) What is the CPU overhead for probing and copying? and (ii) What is the overhead introduced by running LFTJ on individual boxes in comparison to running LFTJ on the whole input data. To answer the first question, we simply run three variants: (a) the full LFTJ, (b) probing and copying data into TrieArraySlices without running LFTJ, and (c) only probing without copying input data nor running LFTJ. Results are shown in the first row of \Figref{FigBoxingEval}. On the X-Axis, we vary the space available for boxing. The individual points range from $5,10,\dots$ up to $200\%$ of the input data size in TrieArray representation. We chose to range up to 200\% since the input is essentially read twice by \lftjtri: once for each of the dimensions $x$ and $y$. 

\textit{Results.} Answering question (i): We can see that the CPU work performed for probing and copying is very low in comparison to the work done by the join evaluation, even when the box sizes are limited to as little as $5\%$ of the size of the input. Answering (ii), we look at the red lines for LFTJ and compare the curve with the value at the far right as this one is achieved by using a single box. The real-world data sets behave as expected: starting at around 25\%, they level out demonstrating that the CPU overhead is low if the available memory is not too much smaller than the input data size. Now, for the synthetic data sets, we see that unexpectedly, using \emph{more} boxes \emph{reduces} the CPU work (memory range 10\%--200\%). We speculate that this is because the boxed version might reduce the work done in binary searches for \texttt{seek} since the space that needs to be searched is smaller. Only at 5\%, does this trend reverse and using more smaller boxes takes longer.

\myparaND{How well does Boxing do with limited memory?} We are also interested in the performance of the boxing technique when disk I/O needs to be performed. Here, we run the same experiments as above but we clear all linux system caches (see \appref{AppendixCaches}) before we start a run. We further use Linux's \emph{cgroup} feature to limit the total amount of RAM used for the program (data+executable) and any caches used by the operating system to buffer I/O on behalf of the program. As actual limit we use the value given to the boxing and shown on the X-Axis plus a fixed 100MB (that accounts for the output buffer and the size of the executable). Results are shown in the second row of \Figref{FigBoxingEval}. We see that probing is still very cheap even for the 5\% memory setting; Provisioning the data now has noticeable costs for low-memory settings (25\% and below). However, even then, it is mostly dominated by the time to actually perform the in-memory joins. This is even more so for the real-world data sets. Overall, with around 25\% or more memory, boxed LFTJ's performance stays constant indicating that I/O is not the bottleneck. For example, we can count all 37 billion triangles in the TWITTER dataset in around 29 minutes without I/O and only need up to 35 minutes with disk I/O.

In the third row of \Figref{FigBoxingEval}, we show number of boxes used as well as the total amount of memory copied for provisioning as a multiple of the TrieArray input size from \Figref{FigTableData}. We see that the number of boxes is generally below 100 unless the memory is restricted to below 25\%; similarly, we never copy more than 15x of the input data even for a 5\% memory restriction. An example for how the boxes were chosen for the TWITTER data set is shown in \Figref{FigBoxingFig}. Each figure shows the front (x-y) plane of the 3-D input space. Darker pixels stand for more data of the represented area. We see that boxes become smaller around the more data-dense areas. See \appref{AppBoxesFig} for more details. 

\begin{figure}[h]
\subfigure[5\%]      {\includegraphics[width=0.159\columnwidth]{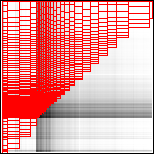}}
\subfigure[10\%]     {\includegraphics[width=0.159\columnwidth]{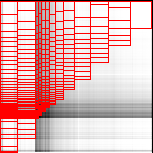}}
\subfigure[25\%]     {\includegraphics[width=0.159\columnwidth]{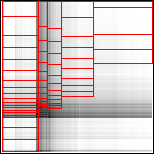}}
\subfigure[35\%]     {\includegraphics[width=0.159\columnwidth]{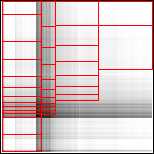}}
\subfigure[75\%]     {\includegraphics[width=0.159\columnwidth]{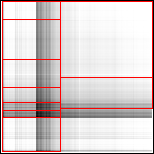}}
\subfigure[100\%]    {\includegraphics[width=0.159\columnwidth]{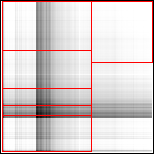}}
\vspace*{-4ex}
\caption{\label{FigBoxingFig} Selected boxes for TWITTER dataset}
\vspace*{-1ex}
\end{figure}

\begin{figure}
\subfigure[Limit: 25\% of input]     {\includegraphics[width=0.48\columnwidth]{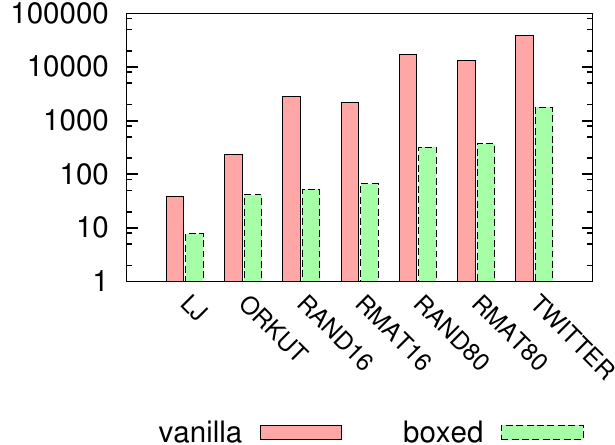}}
\subfigure[Limit: 35\% of input]  {\includegraphics[width=0.48\columnwidth]{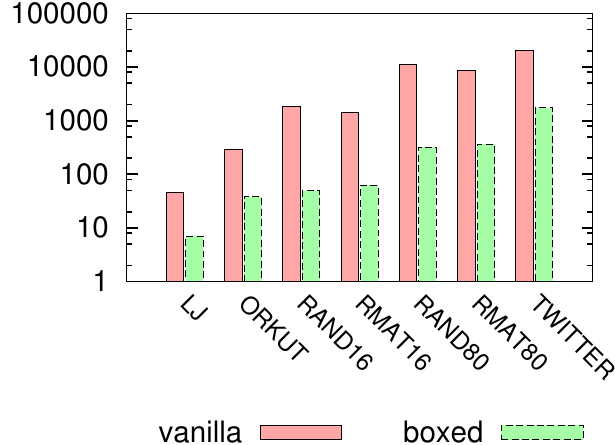}}
\vspace*{-2ex}
\caption{\label{FigVanillaVsBoxed}Vanilla vs.~Boxed LFTJ. \textnormal{For our LFTJ implementation based on TrieArrays. Y-Axis shows wall-clock runtime in seconds. Memory restricted as mentioned.}}
\end{figure}

Last, we are interested in how the boxed LFTJ compares to a variant without our extension. Since LFTJ as presented in \cite{Veldhuizen14} is a family of algorithms that needs to be parameterized by how data is physically stored and how the TrieIterator operations are implemented, answering this question is hard since conclusions for one specific implementation of the data back-end might not hold for another. In particular, our approach of storing data in huge arrays and performing mostly binary searches over them might be particularly bad from an I/O perspective. However, having these considerations in mind, we also ran our version of LFTJ with the \emph{cgroup} memory restrictions and a provisioning mode that does not copy the data but leaves it in memory-mapped files\footnote{We also experimented with this so-called \emph{lazy} provisioning for boxed LFTJ: here, lazy and eager show about the same performance; we omited the data for space reasons.}. The data is thus paged in (from the input file) by the Linux virtual memory system that using a standard replacement strategy. Results for this experiment are shown in \Figref{FigVanillaVsBoxed}: The average speed ratios of vanilla over boxed for the memory levels of 10\%, 25\%, and 35\% are 65x, 30x, and 20x, respectively.

\begin{figure}
\hfill\includegraphics[width=0.35\columnwidth]{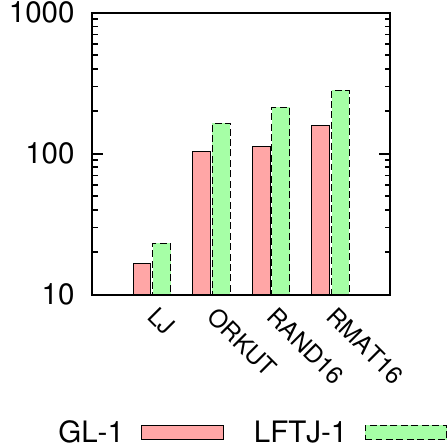}
\hfill\includegraphics[width=0.35\columnwidth]{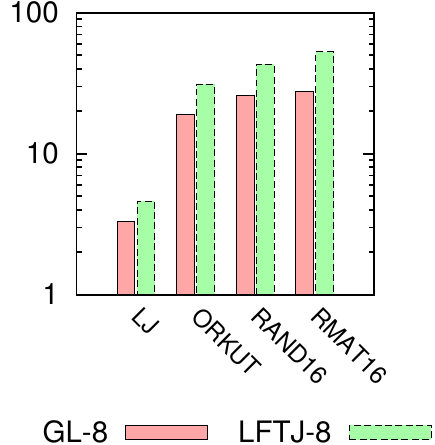}
\hfill\textcolor{white}{.}
\vspace*{-2ex}
\caption{\label{FigCompGraphLab} Performance Graphlab vs.~Boxed LFTJ for single and multi threaded configurations. No resource limitations. Y-axis shows runtime in seconds.}
\end{figure}

\myparaND{How does boxed LFTJ compare to \emph{specialized} best-in-class competitors for triangle listings?} We compare to (1) the triangle counting algorithm presented in Shank's dissertation \cite{schank2007algorithmic} which has been implemented for the graph analysis framework Graphlab \cite{Low2012GraphLab}. We chose this algorithm as our in-memory competitor since it supports multiple threads and was used in other comparisons \cite{Wu2014} before. We also (2) compare to the MGT algorithm by Hu, Tao, and Chung \cite{hu2013massive} as the (to the best of our knowledge) currently best triangle listing algorithm in the out-of-core setting. Our results are shown in \Figref{FigCompGraphLab} and \Figref{FigCompMGT}. The boxed LFTJ is on average 65\% slower than Graphlab, both when run in single-threaded mode as well as in multi-threaded mode with 8 threads. Graphlab, being optimized for an in-memory setting with optional distribution\footnote{which we did not evaluate}, was not able to run any of our large data sets getting ``stuck'' once all of the 32GB of main memory and 32GB of swap space had been consumed.

\begin{figure}
\subfigure[Limit: 10\% of input]{\includegraphics[width=0.48\columnwidth]{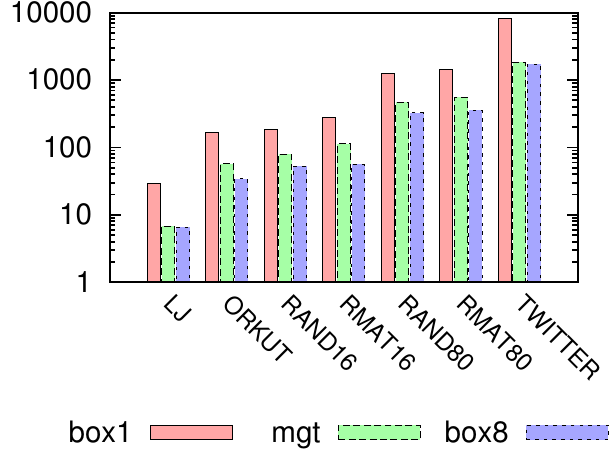}}
\subfigure[Limit: 25\% of input]{\includegraphics[width=0.48\columnwidth]{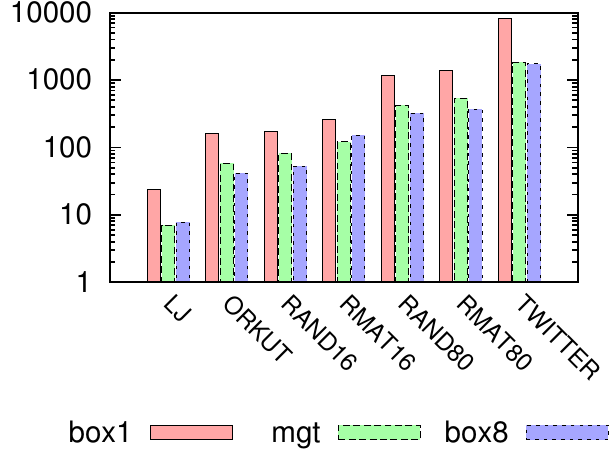}}
\vspace*{-2ex}
\caption{\label{FigCompMGT} Comparison: Boxed LFTJ (1,8 threads) vs. MGT (1 thread) with limited memory. Y-Axis shows wall-clock runtime in seconds.}
\end{figure}

Comparing to MGT (cf.~\Figref{FigCompMGT}): We used the cgroup-memory restrictions and cleaned caches for running MGT and boxed-LFTJ. When we run in single-threaded mode, then MGT outperforms boxed LFTJ by a factor of $3.1$, $2.9$, and $2.9$ in the configurations with 10\%, 25\%, and 35\% of the memory, respectively. Due to time constraints, we did not run LFTJ in single-threaded mode for other configurations. When we allow LFTJ to utilize all of the 4 available cores, we are on average 47\%, 22\%, and 28\%, respectively, faster than the single-threaded MGT. We have not investigated how well MGT parallelizes. Note that MGT internally uses only 32 bits as node identifiers (vs.~our 64bit identifiers). Nevertheless, we used the same values to configure and limit the amount of memory for both MGT and LFTJ.

\section{Related Work} \label{SecRelatedWork}

Related work spans multiple areas at different levels of generality. From most broad to more specific:

The SociaLite effort \cite{seo2013socialite, seo2013distributed} at Stanford also proposes to use systems based on relational joins (in this case Datalog) for graph analysis. They show that declarative methods not only allow for more succinct programs but are also competitive, if not outperform typical other implementations. We did not  compare our join performance with the SociaLite system as it is clearly more feature-rich; it is also Java-based which might or might not influence performance in ways orthogonal to our investigation. We note that the benchmarks presented in \cite{seo2013socialite} and \cite{seo2013distributed} that--among other queries--evaluate counting triangles did not use datasets as large as ours.

A worst-case optimal join algorithm has first been presented by Ngo et al.~in \cite{ngo2012worst} following the AGM bound \cite{atserias2008size} that bounds the maximum number of tuples that can be produced by a conjunctive join. Leapfrog Triejoin by Veldhuizen \cite{Veldhuizen14}, the join algorithm we are using, has been shown to be worst-case optimal as well (modulo a log-factor). In fact, \cite{Veldhuizen14} showed that Leapfrog Triejoin is worst-case optimal (modulo log-factor) for more fine-grained families of inputs. Our work, especially on the worst-case optimality for graphs with limited arboricity was inspired by the worst-case optimal results in \cite{Veldhuizen14}. A good survey and description of this class of worst-case optimal join algorithms is \cite{ngo2014skew}, where the authors not only describe the AGM bound and its application, but also the original NPRR algorithm and LFTJ.

Most recently, Khamis, Ngo, Re, and Rudra proposed so-called \emph{beyond-worst-case-optimal} join algorithms. Here, the performed work is not measured against a worst-case within a set family of inputs---but instead must be proportional to the size of a shortest proof of the results correctness. The idea was proposed by Ngo, Nguyen, Re and Rudra in \cite{ngo2014beyond}. Furthermore, \cite{KhamisNRR14} combines ideas from geometry and resolution transforming the algorithmic problem of computing joins to a geometric one. Following this line of research is very interesting as it might offer even better performance in practice. 

Our boxing approach is most closely related to the classic block-nested loop join (BNLJ)\cite{ramakrishnan2003database}. An interesting avenue for future work would be to investigate how optimizations and results for the BNLJ transfer to the \emph{multi-predicate} LFTJ.

Listing triangles in graphs is a well-researched area in computer science. For the in-memory context, see \cite{latapy2008main} for a recent survey. Triangle listing can also be reduced to matrix multiplication. Recent work that proposes new algorithms based on this approach is \cite{bjoerklund2014}. 
Chiba and Nishizeki \cite{chiba1985arboricity} propose an in-memory triangle listing algorithm that runs in $O(|E|\alpha(G))$ matching the best possible bound we give in \Secref{SecLFTJTriangle}. To the best of our knowledge, our insight that this is the best possible theoretical bound for this class of graphs, is novel and thus provides new insights about these algorithms. Earlier, \cite{Papadimitriou1981131} already showed that enumerating all triangles in planar graphs is a linear-time problem.

Triangle listing in the out-of-core context: Following up on the MGT work \cite{hu2013massive}, Rasmus and Silvestri investigate the I/O complexity of triangle listing \cite{pagh2013input}. They improve on the $I/O$ complexity of MGT from $O(|E|^2/(MB))$ to an expected $O(E^{3/2}/(\sqrt{M}B))$. They also give lower bounds and show that their algorithm is worst-case optimal by proving that any algorithm that enumerates $K$ triangles needs to use at least $\Omega(K/(\sqrt{M}B))$ I/Os. They also give a deterministic algorithm using a color coding technique. Investigating whether the techniques used could be generalized to general joins is a very interesting avenue for future work. Prior to \cite{hu2013massive},  \cite{menegola2010external} proposed an algorithm whith an I/O complexity of $O(|E|+|E|^{1.5}/B)$; furthermore \cite{dementiev2006algorithm} proposed an algorithm with an I/O complexity of $O(|E|^{1.5}/B \cdot \log_{M/B}(|E|/B))$. Cheng et al. \cite{cheng2011finding} study the general problem of finding maximal cliques. We did not benchmark against these algorithms since MGT dominated them by an order of magnitude.

Research has also been done to distribute triangle counting and other graph algorithms \cite{afrati2013enumerating,afrati2010optimizing}, and approaches that use the MapReduce framework \cite{park2014mapreduce,qin2014scalable,suri2011counting,PaghT12}.

\section{Conclusion} \label{SecConclusion}

For the well-studied problem of triangle listing, we have investigated how a \emph{general-purpose} \& worst-case optimal join algorithm compares against \emph{specialized} approaches in the out-of-core context. By using Leapfrog Triejoin, we were able to devise a strategy that not only allows for good theoretical bounds in terms of I/O and CPU costs but we also demonstrated very good performance: For very large input graphs of 1.2 billion edges and more, LFTJ counts triangles with a speed of 4 million input edges per second for uniformly random data; and performs a complete count of the 35 billion triangles in the twitter dataset in little over 25 minutes on a standard 4-core desktop machine \emph{while} limiting the available main memory to around 5GB. Our positive results can be interpreted as a confirmation for the database community's theme of creating systems to empower (domain-expert) users via declarative query interfaces while providing very good performance.

\mypara{Acknowledgements} We thank Ken Ross and Haicheng Wu for comments on an earlier draft; and we thank the anonymous reviewers for their comments. 

\bibliographystyle{abbrv}
\bibliography{references}  %

\balancecolumns
\cleardoublepage
\appendix

\section{Leapfrog Join and\\Leapfrog Triejoin} \label{AppLFTJAlgos}

\begin{algorithm}[H]
\caption{\label{AlgLFJ}Leapfrog join of $n$ unary relations \cite{Veldhuizen14} 
}
\begin{algorithmic}[1]
\Statex \textbf{in:} \ArrayOf \LinIters \Iters{0..n-1} %
\Statex \textbf{variables:} \Int \PVar, \Bool \AtEnd
\Procedure{lfj-init}{\Nothing}
   \If{any iterator in \Iters{0..n-1} is \CallParens{\AtEnd}} 
       \State \AtEnd \Gets \True \Comment{Some input is empty}
   \Else
       \State sort \Iters{0..n-1} increasingly by \CallParens{value}
       \State $\PVar \Gets 0$ ; \AtEnd \Gets \False
       \State \Call{lfj-search}{\,$\!$}
   \EndIf
\EndProcedure
\Procedure{lfj-search}{\Nothing}
   \While{\True}
      \If{\Iters{$\PVar - 1 \mod n$ }.\CallParens{atEnd}}
         \State \AtEnd \Gets \True %
         \State \textbf{return} \Comment{No tuple can be found anymore}
      \EndIf
      \State \MaxVal \Gets \Iters{$\PVar - 1 \mod n$ }.\CallParens{value}
      \State \MinVal \Gets \Iters{\PVar}.\CallParens{value}
      \If{\MinVal = \MaxVal}
         \State \textbf{return} \Comment{Found tuple in intersection}
      \Else
         \State \Iters{\PVar}.\Call{seek}{\MaxVal}
         ; $\PVar \Gets \PVar + 1 \mod n$
      \EndIf
   \EndWhile
\EndProcedure
\Procedure{lfj-next}{\Nothing}
   \State \Iters{\PVar}.\CallParens{next} ; 
   $\PVar \Gets \PVar + 1 \mod n$
   \State \CallParens{lfj-search}
\EndProcedure
\Procedure{lfj-seek}{val}
   \State \Iters{\PVar}.\Call{seek}{val} ; 
   $\PVar \Gets \PVar + 1 \mod n$
   \State \CallParens{lfj-search}
\EndProcedure
\Function{lfj-value}{\Nothing}: 
   \textbf{return} \Iters{0}.\CallParens{value}
\EndFunction
\Function{lfj-atEnd}{\Nothing}:
   \textbf{return} \AtEnd
\EndFunction
\end{algorithmic}
\end{algorithm}
\begin{algorithm}
\caption{\label{AlgLFTJ}Leapfrog Triejoin with $n$ variables \cite{Veldhuizen14}}
\begin{algorithmic}[1]
\Statex \textbf{in:} \ArrayOf \LeapFrogJoinSet \Lfs{1..n} %
\Statex \textbf{variables:} \Int \DVar \Comment{Current depth}
\Procedure{lftj-init}{\Nothing}: $\DVar \Gets 0$
\EndProcedure
\Procedure{lftj-open}{\Nothing}
   \State $\DVar \Gets \DVar + 1$
   \ForAll{$\IterVar \textnormal{ used in } \Lfs{\DVar}$}: \IterVar.\CallParens{open}
   \EndFor
   \State \Lfs{\DVar}.\CallParens{lfj-init}
\EndProcedure
\Procedure{lftj-close}{\Nothing}
   \ForAll{$\IterVar \textnormal{ used in } \Lfs{\DVar}$}: \IterVar.\CallParens{close}
   \EndFor
   \State $\DVar \Gets \DVar - 1$
\EndProcedure
\Procedure{lftj-next}{\Nothing}: \Lfs{\DVar}.\CallParens{lfj-next} \EndProcedure
\Procedure{lftj-seek}{v}: \Lfs{\DVar}.\Call{lfj-seek}{v} \EndProcedure
\Function{lftj-value}{\Nothing}: \textbf{return} \Lfs{\DVar}.\CallParens{lfj-value} \EndFunction
\Function{lftj-atEnd}{\Nothing}: \textbf{return} \Lfs{\DVar}.\CallParens{lfj-atEnd} \EndFunction
\end{algorithmic}
\end{algorithm}

\subsection{TrieIterator Example} \label{AppTrieNavigation}

A TrieIterator is initialized to the root node $r$. Methods for vertical navigation are: \Method{open()} for moving ``down'' to the first children of the current node and \Method{close()} for moving ``up'' to the parent of the current node. Horizontally, movement is restricted to direct siblings, which are accessed via the \emph{LinearIterator} interface that comprises the methods \Method{atEnd}, \Method{next}, \Method{seek}, and \Method{value}.
It is convenient to think of the $l$ children of a node $n$ to be stored increasingly sorted in an array $A$ of size $l$. The methods \Method{bool atEnd()} returns \texttt{true} if the iterator is positioned after the last element (eg., at position $l$). The method \Method{next()} requests to move to the next element; \Method{atEnd} will be true if the iterator was at the last position already (e.g., calling \Method{next()} at position $l-1$). The method \Call{seek}{T $v$} can be used to forward-position the iterator to the element with value $v$; if $v$ is not in $A$, then the iterator is placed at the element with the smallest value $w > v$, or \Method{atEnd} if no such $w$ exists. 
Finally, data is accessed at granularity of a single domain element via the method \Method{T value()}, which returns the element at the current 
position.
The methods \Method{open}, \Method{next}, \Method{seek}, and \Method{value} may only be called if \Method{atEnd()} is \texttt{false}; furthermore, the value $v$ given to \Method{seek(v)} must be at least \Method{value()}; and \Method{value()} must not be called at the root node $r$.

\begin{example}[TrieIterator Navigation] See Figure~\ref{FigTrieTrieArrayB}. The iterator is initially positioned at r; \Method{open()} moves it to a, followed by \Method{next()} to b. Here, \Method{open()} moves to u; \Method{next()} to v; and a \Method{seek(\textnormal{w})} will position the iterator to z since z is the smallest among u,v,z which is larger than w. A call to \Method{next()} causes \Method{atEnd()} to return \texttt{true} after which \Method{close()} would be the only allowed operation, moving the iterator back to b. 
\end{example}

\subsection{Leapfrog TrieJoin Procedure}
Given a join description as a Datalog rule body with $m$ atoms and $n$ variables. For each of the $m$ atoms, a single TrieIterator is created. Furthermore, LFTJ maintains an array of $n$ Leapfrog joins---one join for each variable. The LFJ for variable $x_i$ uses \emph{pointers} to the TrieIterators for atoms that mention the variable $x_i$. %
Overall, LFTJ is implemented as a TrieIterator itself\footnote{The actual join results are collected by walking the Trie.} (see \Algref{AlgLFTJ}). A variable \DVar\ remembers at which level of the output trie the iterator is positioned. The horizontal navigation methods manipulate \DVar, open and close the appropriate TrieIterators, and initialize the Leapfrog joins. The linear iterator methods are then simply delegated to the LFJ which computes the appropriate intersections.

\newpage
\section{Proofs}
\begin{figure}[H]
\hfill
\subfigure[\label{FigThreshingRel}Graph $G$]{
\textcolor{white}{.}$
    \begin{array}{rr}
    \mathbf{E:} \\
    \hline
    \hline
    0, & 24 \\
    1, & 20 \\
    2, & 16 \\
    3, & 12 \\
    4, & 8 \\
    5, & 4 \\
    \hline
    6, & 24 \\
    7, & 20 \\
    .. & ..\\
    \hline
    18, & 24 \\
    .. & ..\\
    23, & 4\\
    \hline
    24, & 24
    \end{array}
$\textcolor{white}{.}
}\hfill
\subfigure[\label{FigThreshingTable}\lftjtri steps when running on $G$]{
\textcolor{white}{.}$
\begin{array}{rrr|rr||r}
a & b & c & D(a) & D(b) & \mathbf{I/O} \\
\hline
\hline
0  &    &    & 24  &      &   \\
0  & 24 &    & 24  & 24   & 1 \\
0  & 24 & 24 & 24  & 24   & 1  \\
\hline                      
1  &    &    & 20  &      &   \\
1  & 20 &    & 20  & 16   & 1 \\
1  & 20 & -  & 20  & 16   & 1  \\
\hline                      
2  &    &    & 16  &      &   \\
2  & 16 &    & 16  &  8   & 1 \\
2  & 16 & -  & 16  &  8   & 1  \\
\hline                      
3  &    &    & 12  &      &   \\
3  & 12 &    & 12  &  24  & 1 \\
3  & 12 & -  & 12  &  24  & 1  \\
\hline
.. & .. & .. & .. & .. 
\end{array}
$
\textcolor{white}{.}
}
\hfill\mbox{}
\vspace*{-2ex}
\caption{\label{FigThreshing} Example input graph that causes \lftjtri to use many I/Os. Parameters: $M=20$, $B=4$ and graph with $N=24$ edges.}
\end{figure}

\subsection{Proof for \propref{PropIOs}} \label{AppProofIOExample}

Consult \Figref{FigThreshingTable}. The variable assignments for $x \defEq a$, $y \defEq b$, and $z \defEq c$ as well as the corresponding neighbors $D(a)$ and $D(b)$ are shown. Each node in $V_1$ causes two block $I/O$s. Further, the block storing the node with id 24 of $V_1$ will be evicted when $x=5$ and $y=4$ or earlier, and the last block with $24$ is thus repeatedly loaded when $x=6$, $x=12$, and $x=18$.

Detailed proof sketch for general case: The outer loop in line 1 of \Algref{AlgFLTJTriangle} ranges from $0$ to $N$. For each value $x \defEq a$, we then join $a$'s neighbors with $V_1$ (line 2) to obtain bindings for $y$. Since each node $a$ has exactly one neighbor $b$, this essentially performs a lookup of $b$ in the first column of $E$. Now, since we spaced the second values in $E$ with a distance of $B$ apart, locating each $b$ within $E$ incurs at least one I/O. Also, since the second values in $E$ repeat in groups of size $T$, the blocks needed for the second group will have been evicted from memory before they are needed, resulting in a single I/O for each tuple in $E$. 
The last step is to intersect the neighbors of $a$ with the neighbors of $b$. In our TrieArray representation, this will incur another I/O.\footnote{When storing relations in B-Trees or as an array of lexicographically sorted tuples the single neighbor of $b$ might already be available once $b$ has been loaded. However, even the reduced I/O cost of at least $|E(G_N)|$ demonstrates thrashing.} $\qed$

\subsection{Proof for Proposition~\ref{propSimpleCPUComplexity}} \label{apppropSimpleCPUComplexity}
The bound on the runtime can easily be obtained from the worst-case optimality wrt.\,input sizes of LFTJ (Corollary~4.3 in \cite{Veldhuizen14}) and the fractional edge-cover bound~\cite{atserias2008size}:
For any three binary relations $R$, $S$, $T$ the result size $|Q|$ of the join 
$Q(x,y,z) \leftarrow R(x,y),S(y,z),T(x,z)$ is limited according to the \emph{fractional edge cover} \cite{atserias2008size}. 
If the sizes of $R$, $S$, and $T$ agree than $|Q|$ is at most $n^{1.5}$ with $n=|R|=|S|=|T|$; adding the log-factor, we obtain the desired bound of $O(\log|E| |E|^{1.5})$.

The complexity is optimal modulo the log-factor since a graph with $|E|$ edges can have $\Omega(|E|^{1.5})$ triangles.
$\qed$

\subsection{Proof for \theoremref{ThmLftjcomplexityalpham}} \label{AppLftjalphabound}

Let $\hat\alpha$ be as required. We now analyze the work done by \lftjtri on a graph $G$ with its directed version $\directed{G}=(V,E)$ (possibly obtained via a $O(|E|\log|E|)$ preprocessing. Let $V_1$ be all nodes in $E$ that have an outgoing edge as usual. It is useful to also consult \Figref{AlgFLTJTriangle} for an explanation of which Leapfrog joins are executed during \lftjtri. We now count the steps at each variable:
\begin{itemize}
  \item At level $x$: We Leapfrog-join $V_1$ with itself yielding a bound of $O(|E|\log|E|)$ based on the requirements for the TrieIterator operations (see \Secref{SecLFTJcomplexityrequirement}).
  \item At level $y$: for each $x \in V_1$, a leapfrog-join is performed between $D(x)$ and $V_1$. As usual, $D(x)$ are the followers of $x$, i.e., $D(x) = \set{y \st (x,y) \in E}$. Summing up all cost and using that
the runtime of a leapfrog-join between two relations of size $s_1$ and $s_2$, respectively, is bound by $O(\log(\max\set{s_1,s_2})\cdot\min\set{s_1,s_2})$, we obtain:
        \[
        \begin{array}{llll}
           & O( & \sum\limits_{x\in V_1} \log(\max\set{d_x,|V_1|}) \min{\set{d_x, |V_1|}} & ) \\
          \subseteq & O( & \log|V| \cdot \sum\limits_{x\in V_1} |D(x)| & )\\
          \subseteq & O( & |E| \log|V| & ) \\
          \subseteq & O( & |E| \log|E| & )
        \end{array}
        \]
  \item At level $z$: Here, for (at most) each edge $(x,y)$ in $E$ we leapfrog join the neighbors of $x$ with the neighbors of $y$. We thus incur the work:
  \[
     \begin{array}{llll}
     & O(           & \sum\limits_{(x,y) \in E} \log(\max\set{d_x,d_y}) \min\set{d_x,d_y} & ) \\
     \subseteq & O( & \sum\limits_{(x,y) \in E}  \log(\max\limits_{v\in V}\set{d_v}) \min\set{d_x,d_y} & ) \\
     \subseteq & O( & \log(\max\limits_{x\in V}\set{d_x}) \cdot \sum\limits_{(x,y) \in E} \min\set{d_x,d_y} & ) \\
     \subseteq & O( & \log|E| \cdot \sum\limits_{(x,y) \in E} \min\set{d_x,d_y} & ) \\
     \end{array}
  \]
  As Lemma 2, Chiba and Nishizeki \cite{chiba1985arboricity} observed that for any graph $G=(V,E)$, the sum $\sum_{(x,y) \in E} \min\set{d_x,d_y}$ is bounded by $2\alpha(G)|E|$. 
  Since $\alpha(G) \leq \hat\alpha(|E|)$ and because $\hat\alpha$ is monotonically increasing, we can bound the work by $O(\log|E|\hat\alpha(|E|)|E|)$, finishing the proof.
\end{itemize}

\subsection{Proof for \theoremref{ThmFineGrainedOptimal}} \label{AppFineGrainedOptimal}

We first show:
\begin{lemma} \label{LemmaNumberoftriangles}
Let $\hat\alpha:\mathbb N \rightarrow \mathbb N^+$ be an arbitrary monotonically increasing, computable function. 
For any $m \in \mathbb N$ there exists a graph with $m$ edges and arboricity at most $\hat\alpha(m)$ with at least 
$\frac{2}{3}m\hat\alpha(m)-\frac{2}{3}m-\frac{4}{3}\hat\alpha(m)^3 -\frac{2}{3}\hat\alpha(m)^2$ triangles.
\end{lemma}
\begin{proof}
\mypara{Informal overview of technique} To get many triangles, we use fully connected graphs $K_k$; to stay under the arboricity limit, we choose $k$ appropriately; to get many edges, we just union many of these $K_k$ into the graph, and then filling up with singleton edges. The math works out to the above quantity.

\mypara{Formal proof} Let $\hat\alpha$ be as required. Fix an $m \in \mathbb N$. 
Let $k = 2\hat\alpha(m)$. Note that the fully connected graphs $K_k$ with $k$ nodes have $l=k(k-1)/2$ edges. 
We construct a graph $G$ by packing as many $K_k$ as we can fit into our ``$m$-edges budget'' and filling up the rest with unconnected edges:
Let $n = \left\lfloor m/l \right\rfloor$, let $G$ be the graph composed of $n$ instances of $K_k$ and $m-nl$ single edges 
not connected to anything else. To complete the proof, we show:
(1) The arboricity of $G$ is $\hat\alpha(m)$, and (2) $G$ has at least 
$\frac{2}{3}m\hat\alpha(m)-\frac{2}{3}m-\frac{4}{3}\hat\alpha(m)^3 -\frac{2}{3}\hat\alpha(m)^2$ triangles.

\textit{Showing (1).} The classic Nash-Williams result \cite{NashWilliams01011964} states that for any graph $G$, its arboricity $\alpha(G)$ is characterized by the maximum edge-node ratio among all its subgraphs:
\[
    \alpha(G) = \max\limits_{S \textnormal{ is subgraph of } G}\left\{ \left\lceil{ \frac{|E(S)|}{|V(S)|-1} } \right\rceil \right\}
\]
It can easily be verified that choosing a $K_k$ as subgraph maximizes the ratio. Thus, $\alpha(G) = \alpha(K_k) = \ceil{k/2} = \hat\alpha(m)$.

\textit{Showing (2)} As short-hand let $\alpha = \hat\alpha(m)$, and let $m' = nl$, which is the largest integer multiple of $l$ that is not larger than $m$. Each $K_k$ has ${k \choose 3 } = k(k-1)(k-2)/6 = l(k-2)/3$ triangles, and we have $n$ of them, totaling in 
\newcommand{\alp}{\alpha}
\[
\begin{array}{lll}
   nl(k-2)/3 & = \frac{1}{3}m'(k-2) \\  
   & \hfill \triangleright \; k = 2\alp \\
   & = \frac{2}{3}m'(\alp-1) \\
   & \hfill \triangleright \; m' \ge m - l + 1 \\
   & \ge \frac{2}{3}(m-l+1)(\alp-1) \\
   \\
   & = \frac{2}{3}(m\alp - m -l\alp + l + \alp - 1) \\
   \\ 
   & = \frac{2}{3}(m\alp - m -l(\alp + 1) + \alp - 1) \\
   & \hfill \triangleright \; l = 2\alp^2 - \alp \\
   & = \frac{2}{3}(m\alp - m - (2\alp^2 - \alp)(\alp + 1) + \alp - 1) \\
   & \\
   & = \frac{2}{3}(m\alp -m - (2\alp^3 + \alp^2 - \alp) + \alp - 1) \\
   & \\
   & = \frac{2}{3}(m\alp -m - 2\alp^3 - \alp^2 + \alp + \alp - 1) \\
   & \\
   & = \frac{2}{3}(m\alp -m - 2\alp^3 - \alp^2 + 2\alp - 1) \\
   & \\
   & = \frac{2}{3}(m\alp -m - 2\alp^3 -(\alp - 1)^2) \\
   & \hfill \triangleright \; \alp \ge 1\\
   & \ge \frac{2}{3}(m\alp -m - 2\alp^3 - \alp^2) \\
   & \\
   & =  \frac{2}{3}m\alp - \frac{2}{3}m - \frac{4}{3}\alp^3 -\frac{2}{3}\alp^2
\end{array}
\]
triangles as required. 
\end{proof}

\newcommand{\specialM}{\ensuremath{m^\star}}
\newcommand{\specialG}{\ensuremath{G^\star}}
\newcommand{\Alg}{\ensuremath{\mathcal A}\xspace}
We proof \theoremref{ThmFineGrainedOptimal} indirect. Let $\hat\alpha: \mathbb N \rightarrow \mathbb N^+$ be an arbitrary, monotonically increasing, computable function, not identical to 1, that is in $o(\sqrt n)$. 
And, let \Alg be an algorithm that lists all triangles in graphs $G=(V,E)$ with $\alpha(G) \leq \hat\alpha(|E|)$ in $o(|E|\hat\alpha(|E|))$ time. Let $T_\alpha(m) : \mathbb N \rightarrow \mathbb N$ be the maximal number of steps \Alg performs on any graph $G=(V,E)$ with $|E| \leq m$ and $\alpha(G) \leq \alpha$. 

Since \Alg runs in $o(|E|\hat\alpha(|E|))$ time: choose $\epsilon_0 = 1/16$ and let $m_0$ be such that for all $m \ge m_0$ we have:
\begin{equation} \label{eqUpperBound}
  T_\alpha(m) \leq \frac{1}{16} m\hat\alpha(m) \quad\quad \textnormal{ for all } m \ge m_0
\end{equation}

From $\hat\alpha \in o(\sqrt n)$: choose $\epsilon_1 = 1/\sqrt{8}$ and let $m_1$ such that for all $m \ge m_1$ we have 
$\hat\alpha(m) \leq \frac{1}{\sqrt{8}} \sqrt m$. 

Now, let $\specialM \in \mathbb N$ be a large enough number such that (1) $\specialM \ge 8$, (2) $\specialM \ge m_0$, (3) $\specialM \ge m_1$, and (4) $\hat\alpha(\specialM) \ge 2$. We can satisfy all conditions since $\alpha$ maps into $\mathbb N^+$, is monotonically increasing, and is not identical to 1. We apply \lemmaref{LemmaNumberoftriangles} with our $\hat\alpha$ for $\specialM$, and conclude there is a graph $\specialG$ with $\specialM$ edges and arboricity at most $\hat\alpha(\specialM)$ with at least $s(\specialM) = \frac{2}{3}\specialM\hat\alpha(\specialM)-\frac{2}{3}\specialM - \frac{4}{3}\hat\alpha(\specialM)^3 - \frac{2}{3}\hat\alpha(\specialM)^2$ triangles. Clearly, \Alg needs to take at least $s(\specialM)$ steps on $\specialG$. Thus:
\[
\begin{array}{llll}
   T_\alpha(\specialM) & \ge & \frac{2}{3}\specialM\hat\alpha(\specialM)-\frac{2}{3}\specialM -\frac{4}{3}\hat\alpha(\specialM)^3  -\frac{2}{3}\hat\alpha(\specialM)^2 \textnormal{\textcolor{white}{xx}}\\
   && \hfill \triangleright \; \hat\alpha(\specialM) \leq \frac{1}{\sqrt{8}} \sqrt{\specialM} \\
   & \ge & \frac{2}{3}\specialM\hat\alpha(\specialM)-\frac{2}{3}\specialM -\frac{1}{6} \hat\alpha(\specialM)\specialM - \frac{1}{12}\specialM \\ \\
   & \ge & \frac{1}{2}\specialM\hat\alpha(\specialM)-\frac{3}{4}\specialM  \\
   && \hfill \triangleright \; \hat\alpha(\specialM) \ge 2 \\ 
   & \ge & \frac{1}{2}\specialM\hat\alpha(\specialM)-\frac{3}{8}\specialM\hat\alpha(\specialM) \\
   && \hfill \\ 
   & \ge & \frac{1}{8}\specialM\hat\alpha(\specialM)  \\
   && \hfill \triangleright \; \specialM \ge 8, \hat\alpha(\specialM) \ge 2 \\ 
   & \ge & \frac{1}{16}\specialM\hat\alpha(\specialM) + 1 \\
   &&  \\ 
   & > & \frac{1}{16}\specialM\hat\alpha(\specialM) \hfill \textnormal{contradicts } \eqref{eqUpperBound} $\qed$ \\
\end{array}
\]

\section{Systems Aspects}

\subsection{Caches and Limiting Resident Memory} \label{AppendixCaches}

\myparaNN To clear Linux file caches we used as root:
\begin{verbatim}
 sync && echo 3 > /proc/sys/vm/drop_caches
\end{verbatim}

\myparaNN We restricted the memory that a process uses for any reason (data, heap, program, caches, etc) using Linux \texttt{cgroups}. Investigating later via top, confirms that only the allowed resident memory is used by the process. We used as root commands such as:
\begin{verbatim}
 # create a group
 mkdir -p /sys/fs/cgroup/memory/limit_mem
 # add process to group
 echo $PID_OF_PROCESS \ 
    > /sys/fs/cgroup/memory/limit_mem/tasks
 # limit memory
 echo $LIMIT_BYTES \
    > /sys/fs/cgroup/memory/limit_mem/\
      memory.limit_in_bytes
\end{verbatim}
\section{More details for Fig.~8} \label{AppBoxesFig}

The input space for \lftjtri is 3-dimensional. We box for \PredsT{1} = $\set{E(x,y),E(x,z)}$ and \PredsT{2} = $\set{E(y,z)}$. Since there were no spills, intervals for dimension $z$ are always $[\NegInf\!\cdot\!\cdot\Inf]$. The figures show how these boxes are created by projecting  the 3-D input space onto the x-y plane. Darker pixels indicate areas where there is more data.
In particular, the image was created as follows: For $E(x,y)$ of the directed graph for the twitter dataset which can be viewed as a point-set in 2D space, create a 2D histogram $H$ with 150x150 bins. Then, because we slice along the first dimension and collect the nodes plus their neighbors, we aggregate over $H$'s second dimension (eg, $y$) values to obtain a 1D histogram $D$ showing the total number of neighbors the nodes in a certain bin have. We then spread this 1D histogram into a 2D space by setting the value at position $x,y$ to $D(x) + D(y)$. This ``image'' is indicative of the total amount of data for a rectangular box. As a last step, we equalize the histogram and map into greyscale to have a prettier picture. The red boxes are then drawn on top according to the made provisioning decisions during the boxing procedure. In the picture the $x$-axis goes from bottom left to bottom right, the $y$ axis from bottom-left to top-left---the same way as in \Figref{FigExampleGraphTiling}. Note, that the number of columns corresponds to how often we need to load the input data at level $y$.

\end{document}